\documentclass[aps,pra,superscriptaddress,preprintnumbers,nofootinbib,longbibliography]{revtex4}

\usepackage{epsfig,graphics}
\usepackage{amsfonts,amssymb,amsmath}            
\usepackage{amsthm}
\usepackage{color}
\usepackage{hyperref}
\usepackage{tikz}
\usetikzlibrary{external}
\newcommand{\comment}[1]{}
\newcommand{\ket}[1]{| #1 \rangle}
\newcommand{\bra}[1]{\langle #1 |}

\newcommand{\ketbra}[2]{|#1\rangle\!\langle#2|}
\newcommand{\proj}[1]{{|#1\rangle\!\langle#1|}}

\newcommand{\id}{\mathbb{I}}

\newcommand{\cA}{\mathcal{A}}

\newcommand{\cD}{\mathcal{D}}
\newcommand{\cE}{\mathcal{E}}
\newcommand{\cH}{\mathcal{H}}
\newcommand{\cI}{\mathcal{I}}

\newcommand{\cM}{\mathcal{M}}
\newcommand{\cP}{\mathcal{P}}
\newcommand{\cR}{\mathcal{R}}
\newcommand{\cS}{\mathcal{S}}

\newcommand{\ot}{\otimes}

\newcommand{\eps}{\epsilon}
\newcommand{\deriv}[1]{\frac{\mathrm{d}}{\mathrm{d}#1}}
\newcommand{\nimplies}{\kern.6em\not\kern -.6em \implies}

\newcommand{\Tr}{{\rm Tr}}
\newcommand{\Tilt}{{\rm Tilt}}

\newcommand{\B}{\textnormal{Box}}

\theoremstyle{plain}
\newtheorem{theorem}{Theorem}
\newtheorem{lemma}[theorem]{Lemma}
\newtheorem{corollary}[theorem]{Corollary}
\newtheorem{proposition}[theorem]{Proposition}

\theoremstyle{definition}
\newtheorem{definition}[theorem]{Definition}
\newtheorem{remark}[theorem]{Remark}

\begin{document}

\title{Keyring models: an approach to steerability}

\date{$9^{\mathrm{th}}$ February 2018}

\author{Carl A.~\surname{Miller}}
\affiliation{National Institute of Standards and Technology, 100 Bureau Dr., Gaithersburg, MD 20899, USA}
\affiliation{Joint Center for Quantum Information and Computer Science, University of Maryland, College Park, MD 20742, USA}
\email{camiller@umd.edu}

\author{Roger \surname{Colbeck}}
\affiliation{Department of Mathematics, University of York,
  York, YO10 5DD, UK}
\email{roger.colbeck@york.ac.uk}

\author{Yaoyun \surname{Shi}}
\affiliation{Aliyun Quantum Laboratory, Alibaba USA, Bellevue, WA 98004, USA}
\email{y.shi@alibaba-inc.com}

\begin{abstract}
  If a measurement is made on one half of a bipartite system, then,
  conditioned on the outcome, the other half has a new reduced
  state. If these reduced states defy classical explanation|that
  is, if shared randomness cannot produce these reduced states for all
  possible measurements|the bipartite state is said to be
  \textit{steerable}. Determining which states are steerable is a
  challenging problem even for low dimensions. In the case of
  two-qubit systems a criterion is known for $T$-states (that is,
  those with maximally mixed marginals) under projective
  measurements. In the current work we introduce the concept of
  \textit{keyring models}|a special class of local hidden state
  models. When the measurements made correspond to real projectors,
  these allow us to study steerability beyond $T$-states.

  Using keyring models, we completely solve the steering problem for
  real projective measurements when the state arises from mixing a
  pure two-qubit state with uniform noise.  We also give a partial
  solution in the case when the uniform noise is replaced by
  independent depolarizing channels.
\end{abstract}

\maketitle

\section{Introduction}

In his 1964 paper~\cite{Bell} John Bell made the fundamental
observation that measurement correlations exhibited by some entangled
quantum states cannot be explained by any local causal model.
Specifically, if $\rho_{AB}$ is the state of a bipartite system shared
by Alice and Bob, and Alice is given a private input
$q\in\mathcal{Q}$ and Bob is given a private input
$s\in\mathcal{S}$, then it is possible for Alice and Bob to measure
$\rho_{AB}$ and produce output messages $a\in\mathcal{A}$ and
$b\in\mathcal{B}$ such that the conditional probability distribution
$\mathbf{P}(ab\mid qs)$ cannot be simulated by any local hidden
variable (LHV) model.

This can be interpreted as a fundamental confirmation of the models
for nonlocality used in quantum physics, and it also has important
applications in information processing.  Device-independent quantum
cryptography is based on the observation that if two untrusted
input-output devices exhibit nonlocal correlations, their internal
processes must be quantum.  With correctly chosen protocols and
mathematical proof, this observation allows a classical user to
manipulate the devices to perform cryptographic tasks and at the
same time verify their security~\cite{Mayers:1998,BHK}.

In 2007, the related notion of quantum steering was
distilled~\cite{Wiseman:2007}, in which, rather than having Bob make a
measurement, we directly consider the subnormalized marginal states
$\tilde{\rho}_B^{\, q,a}$ that he holds when Alice receives input $q$
and produces output $a$.  A local hidden state (LHS) model attempts to
generate these using shared randomness. Denoting the shared randomness
by a random variable $\lambda$, distributed according to probability
distribution $\mu(\lambda)$, Bob can output quantum state
$\sigma_\lambda$, while Alice outputs $a$ according to a probability
distribution $\mathbf{P}_{q,\lambda}(a)$.

Suppose when Alice gets input $q$ she performs a POVM
$\{E^q_a\}_{a\in\cA}$, so that
$\tilde{\rho}_B^{\, q,a}=\Tr_A((E^q_a\ot\id_B)\rho_{AB})$. A LHS model
produces a faithful simulation if
$\tilde{\rho}_B^{\,
  q,a}=\int_\lambda\mathbf{P}_{q,\lambda}(a)\sigma_\lambda\
\mathrm{d}\mu(\lambda)$
for all $q$ and $a$.  If such a model exists, then we say that the
state $\rho_{AB}$ is unsteerable for the family of measurements
$\{\{E^q_a\}_{a\in\cA}\}_{q\in\mathcal{Q}}$.  If a LHS model
exists for all possible measurements Alice could do (i.e., all POVMs),
we say $\rho_{AB}$ is unsteerable.  Conversely, if there exists a set
of measurements for which no LHS model exists, then $\rho_{AB}$ is
said to be steerable.

One can think of steering as an analog of non-locality for the case
where one party (Bob) trusts his measurement device (and hence in
principle could do tomography to determine his marginal state after
being told Alice's measurement and outcome).  It is hence a useful
intermediate between entanglement witnessing (both measurement devices
trusted) and Bell violations (neither trusted) and has applications
such as one-sided device-independent quantum
cryptography~\cite{Branciard:2012} and (sub)channel
discrimination~\cite{Piani:2015}.  Exhibiting new steerable states
offers an expanded toolbox for such problems.

The steering decision problem is to determine whether or not a given
state is steerable.  This problem has proved to be difficult even for
$2$-qubit systems.  To understand why this is so, consider a two-qubit
state $\rho_{AB}$.  If Alice were to measure $\{\proj{0},\proj{1}\}$
on input $q=0$ and $\{\proj{+},\proj{-}\}$ on input $q=1$ (where
$\ket{\pm}=(\ket{0}\pm\ket{1})/\sqrt{2}$), then it is possible for Bob
to obtain one of four subnormalized states which we denote
$\tilde{\rho}_B^{\, 0}, \tilde{\rho}_B^{\, 1}, \tilde{\rho}_B^{\, +},
\tilde{\rho}_B^{\, -}$
(where, for example,
$\tilde{\rho}_B^{\, 0}=\Tr_A[(\proj{0}\ot\id_B)\rho]$).  Determining
whether a LHS model exists for these four states is a search over a
finite-dimensional space and is not difficult
(see~\cite{Cavalcanti,Hirsch} for techniques for searching for LHS
models).  Next suppose Alice additionally performs the measurement
$\{\proj{\pi/4},\proj{5\pi/4}\}$ for input $q=2$, where
\begin{equation}\label{eq:theta}
\ket{\theta}:=\cos\frac{\theta}{2}\ket{0}+\sin\frac{\theta}{2}\ket{1}\, ,
\end{equation}
leading to states $\tilde{\rho}_B^{\, \pi/4},\tilde{\rho}_B^{\, 5\pi/4}$.
There is no guarantee that a local hidden state model that simulates
the previous four states will simulate this new pair as well
(generally, the states
$\tilde{\rho}_B^{\, \pi/8},\tilde{\rho}_B^{\, 5\pi/8}$ are not in the convex
hull of the former states).  A new search for local hidden state
models is required, and the search space increases exponentially with
each new measurement.  Thus a direct approach|even when just
dealing with measurements of the form
$\{\proj{\theta},\proj{\theta+\pi}\}$|is unlikely to be feasible.

Previous work on steering has achieved success by exploiting the
symmetries of certain classes of states.  For the class of Werner
states~\cite{Werner}
$\left\{\rho_{AB}(\eta)\mid\eta\in[0,1]\right\}$ given by
\begin{eqnarray}\label{eq:werner}
\rho_{AB}(\eta)&=&\eta\proj{\Phi_+}+(1-\eta)\id/4,
\end{eqnarray}
where $\ket{\Phi_+}=\frac{1}{\sqrt{2}}(\ket{00}+\ket{11})$, an exact
classification of $P$-steerability (i.e., steerability for all
projective measurements) has been performed (see
Appendix~\ref{wernersubsec} for a summary of results on Werner
states).  More recently a complete classification of $P$-steerability
for $T$-states (i.e., states for which $\rho_A$ and $\rho_B$ are
maximally mixed) has been
given~\cite{Jevtic:2015,Nguyen:PRA,Nguyen:2016}. [Note that the
requirement on $\rho_B$ can be dropped|see Lemma~\ref{lem:M} below.]
In both cases the methods depend critically on the symmetry of the
states.  For $2$-qubit states outside the family of $T$-states,
partial results on steerability exist
(e.g.,~\cite{Jones07,Bowles:2016}) but a full classification is not
known.

In the current work, we develop new techniques to decide steerability
in the case where $\rho_A$ is not maximally mixed.  We study Real
Projective (RP)-steerability (i.e., steerability by the family of all
measurements of the form $\{\proj{\theta},\proj{\theta+\pi} \}$) for
real two-qubit states. [A state is \emph{real} if its matrix
  elements are real in the $\{\ket{0},\ket{1}\}$ basis.]  To
illustrate our techniques, we give a complete classification of
RP-steerability for the class of states,
$\left\{\rho_{AB}(\alpha,\eta)\right\}$, formed by mixing partially
entangled pure states with uniform noise, i.e.,
\begin{eqnarray}
\rho_{AB}(\alpha,\eta)=\eta\proj{\phi_\alpha}+(1-\eta)\id/4\, ,
\end{eqnarray}
where $\ket{\phi_\alpha}=\cos\alpha\ket{00}+\sin\alpha\ket{11}$.  The
classification is shown in Figure~\ref{introfig}, where the
shaded/unshaded region represents the states that are
unsteerable/steerable for real projective measurements. As a special
case we recover the existing result~\cite{Jones11,Uola16} that Werner
states are RP-steerable if and only if $\eta>2/\pi$ (see
Theorem~\ref{thm:Werner}).

Our criterion also applies to a larger class of real $2$-qubit
states, specifically, all states whose steering ellipse is tilted at
an angle less than $\pi/4$|see Theorem~\ref{steerthm2} and
Corollary~\ref{cor:main} for the formal statements.  To achieve this
classification we introduce the concept of \textit{keyring models},
which are a geometrically motivated class of local hidden state models
for one-dimensional families of measurements.  We explain these in
more detail in the next subsection.

Our approach invites generalizations.  In its current form we have a
criterion for steerability among all real $2$-qubit states whose
steering ellipse is tilted at an angle less than $\pi/4$.  With
additional work one may be able to go further identify the set of all
$RP$-steerable real $2$-qubit states.  Additionally, the keyring
approach could be applied in more general scenarios where steering is
attempted with any one-dimensional family of measurements.

\begin{figure}[h]
\begin{center}
\includegraphics[width=0.4\textwidth]{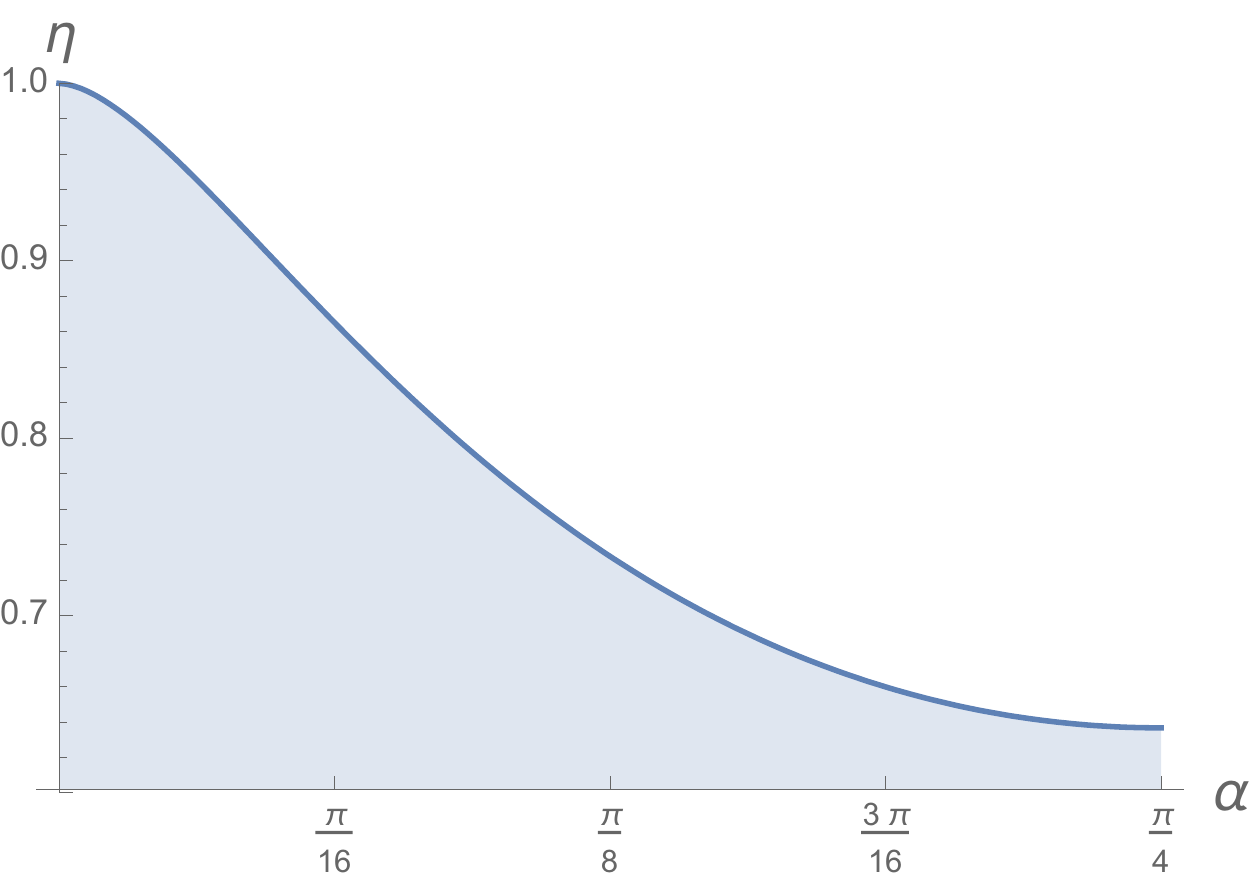}
\end{center}
\caption{$RP$-unsteerability for the states
  $\rho_{AB}(\alpha,\eta)$. In the shaded region the states are
  unsteerable under real projective measurements, while above it they
  are steerable. (Note that the shaded region extends to $\eta=0$.)}
\label{introfig}
\end{figure}

Studying the behavior of qubit states under real projective
measurements is a natural problem for experimental setups in which
measurements in one plane of the Bloch sphere are easier than the most
general measurements. However, another future goal would be to extend
our methods to arbitrary complex measurements on $2$-qubit states.
This looks more challenging|steering with a $2$-dimensional family
of measurements is considerably harder than with a $1$-dimensional
family of measurements|but if it can be accomplished, it would be
an important step towards a complete criterion for steering among
arbitrary $2$-qubit states.

Keyring models can also be used to construct a class of LHV models if
we also use a (classical) function on Bob's side to map his
input and the hidden state to his output. They can hence be applied to
the related problem of classically simulating bipartite correlations
and may, for example, be useful for shedding new light on the problem
of identifying the smallest detector efficiency for observing Bell
inequality violations.  We hope to find further applications of
keyring models in this direction.

\subsection{Sketch of the proof techniques}
The difficulty in establishing steerability over all measurements is
the need to rule out \emph{all} LHS models.  Our proof begins with the
observation that, in the case where $\rho_{AB}$ is a real $2$-qubit
state and where the set of measurements comprises real projective
measurements (i.e., those of the form
$\{\proj{\theta},\proj{\theta+\pi}\}$), a more tractable (though still
infinite dimensional) class of LHS models suffices.  Specifically, we
consider a class of LHS models that we call ``keyring models'', which
we now define.

Let $\mathbb{RP}^1$ denote the set of all real one-dimensional
projectors on $\mathbb{C}^2$ (i.e., the set
$\{\proj{\theta}\}_{\theta\in[0,2\pi)}$).  A keyring model is a pair
$(\mu,\{f_\theta\}_\theta)$, where $\mu$ is a probability distribution
on $\mathbb{RP}^1$, and $f_\theta\colon\mathbb{RP}^1\to[0,1]$ is a
two-step function|that is, roughly speaking, a function that takes
two possible values and switches between them at two elements of
$\mathbb{RP}^1$ (see Definition~\ref{twostepdef}).  The word
``keyring'' refers to the configuration of the two switching points on
$\mathbb{RP}^1$ as $\theta$ varies.  An example configuration is shown
in Figure~\ref{keyringfig}.  (This definition is related to the local
hidden state models
of~\cite{Jones11,Jevtic:2015,Nguyen:PRA,Nguyen:2016}, which are based
on functions on $\mathbb{RP}^2$ that are supported on half-spheres.
One key difference in the definition of a keyring model is that there
is no uniformity in the positioning of the switching points of the
functions $f_\theta$|they need not be diametrically opposite.)

\begin{figure}[h]
\begin{center}
\includegraphics[width=0.38\textwidth]{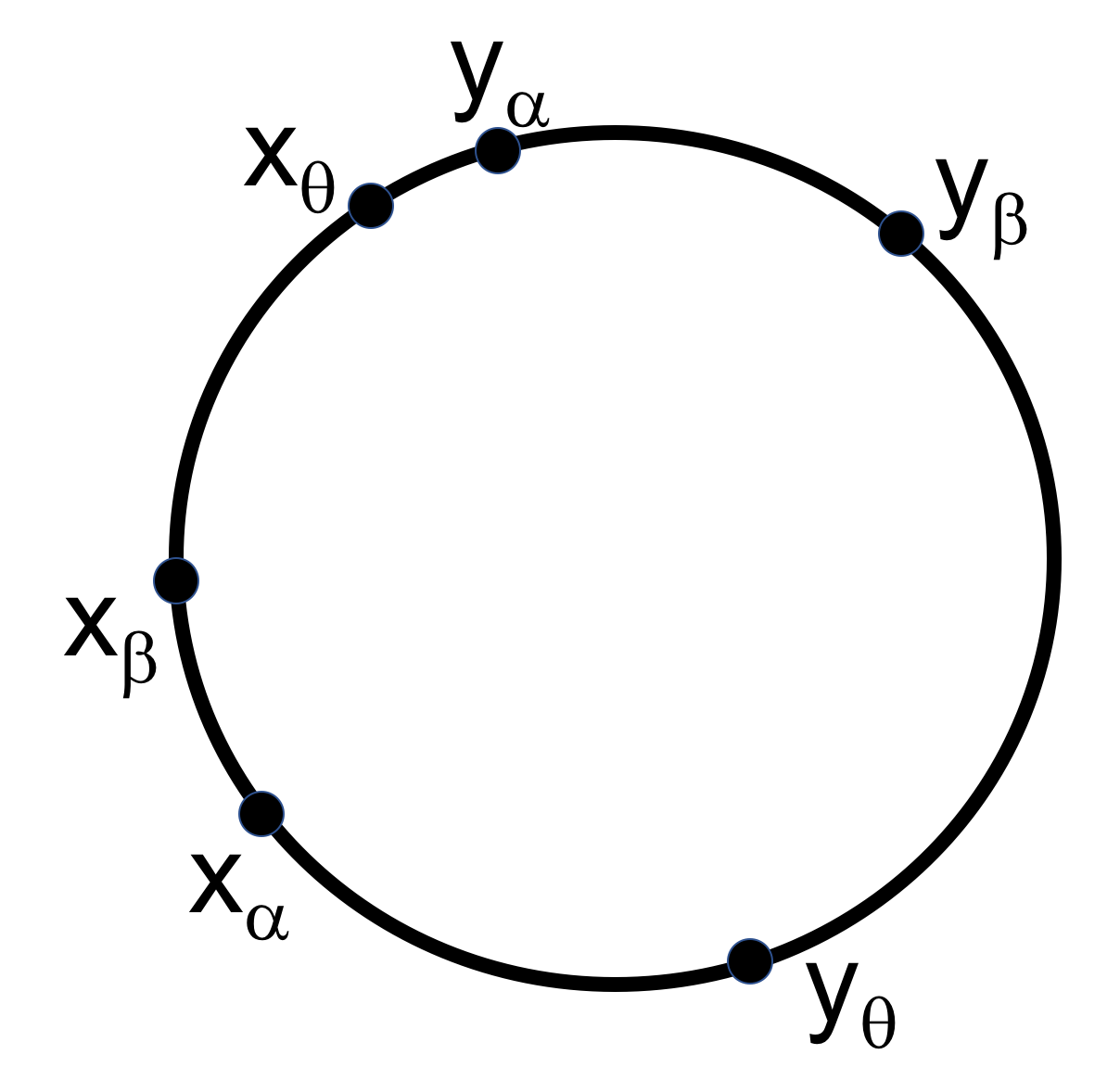}
\end{center}
\caption{An example configuration of endpoints in a keyring
  distribution.  For every point in $\mathbb{RP}^1$ there are
  two associated endpoints in $\mathbb{RP}^1$.  Here three pairs of
  endpoints are illustrated with $(x_{\alpha},y_{\alpha})$ being the
  end points for $\alpha$, for example.}
\label{keyringfig}
\end{figure}

We show that $\rho_{AB}$ is RP-steerable if and only if
it can be simulated by a keyring model.  Denoting the subnormalized
reduced states on Bob's side by
$\tilde{\rho}_B(\theta)=\Tr_A\left((\proj{\theta}\ot\id_B)\rho_{AB}\right)$,
this is equivalent to the requirement
\begin{eqnarray}
\tilde{\rho}_B(\theta)&=&\int_{x\in\mathbb{RP}^1}xf_\theta(x)\ \mathrm{d}\mu
\end{eqnarray}
for all $\theta$.
From this we can conclude that that if the circumference of the
steering ellipse $\{\tilde{\rho}_B(\theta)\}$ is greater than $2$, i.e.,
\begin{eqnarray}
\label{circum1}
\int_{\mathbb{RP}^1}\left\|\deriv{\theta}\tilde{\rho}_B(\theta)\right\|_1\mathrm{d}\theta&>&2,
\end{eqnarray}
where $\|\cdot\|_1$ is the trace norm, then the state $\rho_{AB}$ has no local hidden state model
(see Proposition~\ref{prop:circum}).

At this point our proof diverges from that
of~\cite{Jones11,Jevtic:2015,Nguyen:PRA,Nguyen:2016}, since the converse of the above
statement is not true in our case: if (\ref{circum1}) fails to hold,
there could still be no local hidden state model.  However, the
following stronger condition guarantees the existence of a local
hidden state model:
\begin{eqnarray}
\label{circum2}
\int_{\mathbb{RP}^1}\left|\deriv{\theta}\tilde{\rho}_B(\theta)\right|\mathrm{d}\theta&\leq&2\rho_B,
\end{eqnarray}
where $|X|=\sqrt{X^\dagger X}$ is the absolute value of the operator.
Moreover, the state $\rho_{AB}$ is steerable if and only if
\begin{eqnarray}
\label{rhoprime}
\rho'_{AB}:=(\id_A\ot Y)\rho_{AB}(\id_A\ot Y)
\end{eqnarray}
is steerable for all positive definite $Y$ (see Lemma~\ref{lem:M}),
and by substituting in $\rho'_{AB}$ for $\rho_{AB}$ in (\ref{circum1})
and (\ref{circum2}) we obtain an infinite family of criteria for
$RP$-steerability and $RP$-unsteerability.  We thus need to find a $Y$
such that one of (\ref{circum1}) and (\ref{circum2}) holds for
$\rho'_{AB}$.

The most technically difficult part of our proof then shows that there
must exist a positive definite density matrix $Y$ such that
\begin{eqnarray}
\label{compels}
Y^{-1}\left[\int_{\mathbb{RP}^1}\left|Y\left(\deriv{\theta}\tilde{\rho}_B(\theta)\right)Y\right|\mathrm{d}\theta\right]Y^{-1}
\end{eqnarray}
is a scalar multiple of $\rho_B$.  This compels~\eqref{compels} to
either be greater than, or less than or equal to $\rho_B$, and thus we
achieve a criterion for steering which is both necessary and
sufficient.  We prove this by demonstrating that if we let $Y$ tend to
any projector $P$ in $\mathbb{RP}^1$, then~\eqref{compels} must tend
to an operator proportional to the orthogonal projector $\widehat{P}$.
Any continuous map from a $2$-dimensional disc to itself which rotates
the boundary of the disc must be an onto map, and this gives the
desired result.  (The proof of the aforementioned limit assertion is
surprisingly subtle|it turns out that the rate at which the
normalization of~\eqref{compels} approaches $\widehat{P}$ is only
logarithmic.)

Theorem~\ref{steerthm2} gives a formal statement of our main
result. To apply the criteria (e.g., to obtain Figure~\ref{introfig}),
we use numerical computations to find the appropriate operators $Y$
from a given state $\rho_{AB}$.

\section{Preliminaries}

\subsection{Notation and Definitions}

For any Hilbert space $\cH$, let $\mathcal{A}(\cH)$ denote the set of
all Hermitian operators on $\cH$, $\mathcal{P}_\geq(\cH)$ be the set
of positive semidefinite operators on $\cH$, $\mathcal{P}_>(\cH)$ be
the set of positive definite operators on $\cH$, $\mathcal{D}(\cH)$
denote the set of all density operators on $\cH$, and
$\mathcal{D}_>(\cH)$ denote the set of all positive definite density
operators on $\cH$. Let $\mathcal{RA}(\cH),\mathcal{RP}_\geq(\cH)$
etc.\ denote the respective subsets of real operators (an operator $X$
is real if $\bra{i}X\ket{j}\in\mathbb{R}$ for all $i,j$, where
$\{\ket{i}\}$ is the standard basis).  If
$A,B\in\mathcal{P}_{\geq}(\cH)$ we write $A\geq B$ to mean that
$A-B\in\cP_{\geq}(\cH)$ and $A\ngeq B$ for the complement of this.
For an operator $X$ on $\cH$ we use $|X|:=\sqrt{X^\dagger X}$ and
$\|X\|_1:=\Tr|X|$, the latter being the \emph{trace norm} of $X$. If
$\Tr(X)\neq 0$, we use $\left<X\right>$ to denote the normalized
version of $X$, i.e., $\left<X\right>:=X/\Tr(X)$. In addition, if $Y$
is also an operator on $\cH$, then we use
$\langle X,Y\rangle:=\Tr(X^\dagger Y)$.

Throughout this paper, we take
$\mathcal{H}_A=\mathcal{H}_B=\mathbb{C}^2$ to be qubit systems
possessed by Alice and Bob and use
$\mathbb{RP}^1\subseteq\mathcal{RD}(\cH)$ to denote the set of
one-dimensional real projectors on $\mathbb{C}^2$.

\subsubsection{The steering ellipse}
Any operator $\lambda\in\mathcal{RA}(\mathbb{C}^2)$ can be
expressed uniquely in terms of real numbers $n$, $r_1$, $r_3$ as 
\begin{eqnarray}
\lambda=\frac{1}{2}(n\id+r_1\sigma_1+r_3\sigma_3)\, ,
\end{eqnarray}
where $\sigma_1=\ketbra{0}{1}+\ketbra{1}{0}$ and
$\sigma_3=\proj{0}-\proj{1}$ are the usual Pauli operators. Note that
$\lambda\in\mathbb{RP}^1$ if and only if $n=1$ and $r_1^2+r_3^2=1$.

The \emph{tilt} of $\lambda$, denoted $\Tilt(\lambda)$, is the
quantity $\sqrt{r^2_1+r^2_3}/|n|$ (if $n=0$, the tilt is $\infty$).
The \emph{tilt angle} of $\lambda$ is $\arctan (\Tilt(\lambda))$.  If
we think of $(n,r_1,r_3)$ as $3$-dimensional Cartesian coordinates,
then the tilt angle of $\lambda$ is angle that it forms with the
$(1,0,0)$ axis.  We use these coordinates when we sketch steering
ellipses later in this work. Note that an operator is positive
semidefinite if and only if $n\geq 0$ and its tilt is less than or
equal to $1$.  It is useful to note that
\begin{equation}\label{eq:1norm}
\|\lambda\|_1=\left\{\begin{array}{ll}|n|&\text{if }\Tilt(\lambda)\leq 1\\
\sqrt{r_1^2+r_3^2}&\text{if }\Tilt(\lambda)>1\end{array}\right.
\end{equation}

Let $\rho_{AB}\in\mathcal{RD}(\mathbb{C}^2\ot\mathbb{C}^2)$.  Then, the \textit{steering ellipse}
of $\rho_{AB}$ on $B$ is the function $\tilde{\rho}_B\colon\mathbb{RP}^1\to\mathcal{P}_\geq(\mathbb{C}^2)$ given by
\begin{eqnarray}
  \tilde{\rho}_B(\theta)&:=&\Tr_A\left[\left(\proj{\theta}\ot\id_B\right)\rho_{AB}\right]\, ,
\end{eqnarray}
where $\ket{\theta}$ is defined in~\eqref{eq:theta}.  Note that
$\{\ket{\theta},\ket{\theta+\pi}\}$ form an orthonormal basis, so
$\rho_B=\tilde{\rho}_B(\theta)+\tilde{\rho}_B(\theta+\pi)$ for any
$\theta$.  (In the more general case of arbitrary projective
measurements, the states on Bob's side are a two-parameter family that
define an ellipsoid rather than an ellipse. Note also that the term
``steering ellipsoid'' is used to refer to the set of normalized
states in the literature~\cite{JPJR,Jevtic:2015}, while our steering
ellipse comprises subnormalized states).
\begin{definition}
  Let $\rho_{AB}\in\mathcal{RD}(\mathbb{C}^2\ot\mathbb{C}^2)$.  Then, the
  \emph{tilt of the steering ellipse of $\rho_{AB}$} is the equal to
  the tilt of any nonzero vector that is normal to the $2$-dimensional
  affine space that contains the steering ellipse of
  $\rho_{AB}$. [If the steering ellipse does not span a
    $2$-dimensional affine space (i.e., it is degenerate) then we say
    that its tilt is equal to $\infty$.]
\end{definition}
Note that if the tilt of the steering ellipse is less than or equal to
$1$, then no element of the steering ellipse is strictly greater (in
the positive semidefinite sense) than any other.  This is a
consequence of Lemma~\ref{lem:perptilt} in the appendix.

\subsubsection{Local hidden state models}
In this section we give a definition of a local hidden state model. It
is not the most general definition possible, but it suffices for our
purposes because of the form of the steering problem we are
considering, as we now explain.

In the most general sense, a local hidden state model for a set of
real $2$-qubit subnormalized states
$\left\{\tilde{\rho}_B^{\, q,a}\right\}_{q\in\mathcal{Q},a\in\cA}$ is
a probability distribution $\mu$ on $\mathcal{D}(\mathbb{C}^2)$ and
set of functions
$\left\{f_{q,a}\colon\mathcal{D}(\mathbb{C}^2)\to[0,1]\right\}_{q,a}$
with $\sum_{a\in\cA}f_{q,a}(x)=1$ such that
\begin{eqnarray}
\tilde{\rho}_B^{\, q,a}&=&\int_{x\in\mathcal{D}(\mathbb{C}^2)}xf_{q,a}(x)\ \mathrm{d}\mu\, .
\end{eqnarray}
(To connect with the earlier description, $f_{q,a}(x)$ is the probability that Alice gives the outcome $a$ for measurement $q$ when the hidden variable takes the value $x$.)
However, via the map $\mathcal{D}(\mathbb{C}^2)\to\mathcal{RD}(\mathbb{C}^2)$ given
by $x\mapsto(x+\overline{x})/2$, we may assume $\mu,f_{a,q}$
have support $\mathcal{RD}(\mathbb{C}^2)$, and by decomposing each operator in $\mathcal{RD}(\mathbb{C}^2)$
into a convex combination of one-dimensional projectors, we may further assume
that $\mu,f_{q,a}$ have support $\mathbb{RP}^1$.  We are thus led to the following definition.
\begin{definition}
  A \emph{local hidden state model} for a set
  $\left\{\tilde{\rho}_B^{\,
      q,a}\right\}_{q\in\mathcal{Q},a\in\cA}\subseteq\mathcal{RP}_\geq(\mathbb{C}^2)$
  is a pair $(\mu,\{f_{q,a}\}_{q,a})$ such that $\mu$ is a probability
  distribution on $\mathbb{RP}^1$,
  $f_{q,a}\colon\mathbb{RP}^1\to[0,1]$ with
  $\sum_{a\in\cA}f_{q,a}(x)=1$ for all $q$, and
\begin{eqnarray}
\tilde{\rho}_B^{\, q,a}&=&\int_{x\in\mathbb{RP}^1}xf_{q,a}(x)\ \mathrm{d}\mu\, .
\end{eqnarray}
\end{definition}
In the case of steering for real 2-qubit states under real projective
measurements, it suffices to consider whether we can find
$(\mu,\{f_\theta\}_\theta)$ with
$f_\theta\colon\mathbb{RP}^1\to[0,1]$ such that
\begin{equation}
\tilde{\rho}_B(\theta)=\int_{x\in\mathbb{RP}^1}xf_\theta(x)\ \mathrm{d}\mu\, .
\end{equation}
(Here $f_{\theta}(x)$ is the probability that Alice gives the outcome
corresponding to the first projector for the measurement
$\{\proj{\theta},\proj{\theta+\pi}\}$ when the hidden variable has
value $x$.) If such a $(\mu, \{f_\theta\}_\theta)$ can be found, this
constitutes a LHS model for the set
$\{\tilde{\rho}_B(\theta)\}_{\theta\in[0,2\pi)}$ and we say that
$\rho_{AB}$ is \emph{RP-unsteerable}.  Conversely, if no such model
exists, we say that $\rho_{AB}$ is \emph{RP-steerable}.

\begin{remark}\label{rmk:convex}
  The property of having a LHS model is convex, i.e., if $\rho_{AB}$
  and $\rho'_{AB}$ have LHS models (for some set of measurements), then so does
  $p\rho_{AB}+(1-p)\rho'_{AB}$ for all $0\leq p\leq 1$ (and the same
  set of measurements).
\end{remark}

\section{Keyring models}

In this section we formalize the class of keyring models.  We begin with some
preliminary definitions.
Drawing from~\cite{Nguyen:2016}, if $\mu$ is a probability distribution
on $\mathbb{RP}^1$, let
$\B(\mu)$ denote the convex set of all operators of the form
\begin{eqnarray}
\int_{x\in\mathbb{RP}^1}xf(x)\ \mathrm{d}\mu,
\end{eqnarray}
where $f:\mathbb{RP}^1\to[0,1]$.  Note that
$\B(\mu)\subset\cR\cA(\mathbb{C}^2)$ with $\Tr(z)\leq 1$ for
$z\in\B(\mu)$ and that an ellipse has a local hidden state model if
and only if it is contained in $\B ( \mu )$ for some probability
distribution $\mu$.

Note that there is a natural identification between $\mathbb{RP}^1$
and the unit circle $S^1\subseteq \mathbb{R}^2$ which is given by
$\frac{1}{2}(\id+r_1\sigma_1+r_3\sigma_3)\leftrightarrow(r_1,r_3)$
with $r_1^2+r_3^2=1$.  We say that a sequence
$s_1,s_2,s_3\in \mathbb{RP}^1$ is a \textit{clockwise} sequence if
the images of $s_1,s_2,s_3$ form a clockwise sequence in $S^1$, and
a \textit{counterclockwise} sequence if the images of $s_1,s_2,s_3$
form a counterclockwise sequence in $S^1$.  (If any of the points
$s_1,s_2,s_3$ are the same, then we will say that the sequence is
both clockwise and counterclockwise.)  We say that a sequence
$t_1,\ldots,t_n\in\mathbb{RP}^1$ is clockwise (resp.\
counterclockwise) if every $3$-term subsequence of
$t_1,t_2,\ldots,t_n,t_1$ is clockwise (resp.\ counterclockwise).

For any $x,y\in\mathbb{RP}^1$, let $[x,y]$ denote the set of all
$z\in\mathbb{RP}^1$ such that $x,y,z$ is a clockwise sequence.
Let $(x,y)=\mathbb{RP}^1\smallsetminus[y,x]$.  Note that, as
implied by the notation, $[x,y]$ is a closed set and $(x,y)$ is open.

\begin{definition}
\label{twostepdef}
A function $f\colon\mathbb{RP}^1\to[0,1]$ is a \emph{two-step
function} if there are (not necessarily distinct) elements
$x,y\in\mathbb{RP}^1$ and $q\in[0,1/2]$ such that
\begin{eqnarray}
f(z)&=&\left\{\begin{array}{ccl}1-q&\textnormal{ if }
&z\in(x,y)\\
q&\textnormal{ if }&z\in(y,x),
\end{array} \right.
\end{eqnarray}
with $q\leq f(x)\leq 1-q$ and $q\leq f(y)\leq 1-q$.
We refer to $q$ as the \emph{bias} of the function and to $x,y$ as the
\emph{endpoints} of the function.  If $q<1/2$, then we refer
specifically to $x$ as the \emph{left endpoint} and to $y$ as the
\emph{right endpoint}.
\end{definition}

A \textit{keyring} model for a set
$\{\sigma^a\}_a\subseteq\mathcal{RP}_\geq(\mathbb{C}^2)$ of
subnormalized states is a local hidden state model $(\mu,\{f_a\})$ in
which the functions are all two-step functions (see
Figure~\ref{keyringfig}).  The next proposition, which is proven in
Appendix~\ref{app:2prop}, shows that any set that has a local hidden
state model also has a keyring model.  Hence when considering our
steering problem it suffices to restrict the set of local hidden state
models to keyring models.

\begin{proposition}
\label{2stepprop}
Let $\mu$ be a probability distribution on $\mathbb{RP}^1$.  Any
element of $z\in\B(\mu)$ can be written
\begin{eqnarray}
z=\int_{x\in\mathbb{RP}^1}xg(x)\ \mathrm{d}\mu ,
\end{eqnarray}
where $g$ is a two-step function.
If $z$ is on the boundary of $\B(\mu)$, then such a function $g$
exists with bias $q=0$.
\end{proposition}

Next we will use these techniques to prove a geometric fact
about steerability.  Let us say that the \textit{length} of a
piecewise differentiable curve
$S\colon[0,1]\to\mathcal{RA}(\mathbb{C}^2)$ is its length
under the trace norm:
\begin{eqnarray}
\int_0^1 \left\|\deriv{t}S(t)\right\|_1\mathrm{d}t.
\end{eqnarray}

\begin{proposition}
\label{prop:circum}
Let $\rho_{AB}\in\mathcal{RD}(\mathbb{C}^2\ot\mathbb{C}^2)$ be
a two-qubit state whose steering ellipse has tilt $<1$ and whose
steering ellipse $\{\tilde{\rho}_B(\theta)\}_\theta$ has a local
hidden state model.  Then, the length of
$\{\tilde{\rho}_B(\theta)\}_\theta$ is no more than $2$.
\end{proposition}
Note that, using~\eqref{eq:1norm}, the length of this curve is the
Euclidean length of the projection of the ellipse onto the $n=0$ plane
in Bloch representation.  It can be calculated using
$$\int_0^{2\pi}\sqrt{\left(\deriv{\theta}r_1(\theta)\right)^2+\left(\deriv{\theta}r_3(\theta)\right)^2}\
\mathrm{d}\theta\, .$$

To prove Proposition~\ref{prop:circum}, we first consider LHS models
in which the distribution $\mu$ is supported on a finite set of points
of $\mathbb{RP}^1$.  Any probability distribution $\mu$ on
$\mathbb{RP}^1$ can be approximated to an arbitrary degree of accuracy
by a probability distribution which is supported on a finite set of
points in the sense that for any $\eps>0$ there exists a finitely
supported distribution $\mu'$ such that for all two-step functions $f$
we have
$$\left\|\int_{x\in\mathbb{RP}^1}xf(x)\ \mathrm{d}\mu-\int_{x\in\mathbb{RP}^1}xf(x)\ \mathrm{d}\mu'\right\|_1\leq\eps\, .$$

The next lemma shows that if $\mu$ is a probability distribution with
finite support then certain slices of $\B(\mu)$ must have
circumference $\leq 2$ under the trace norm.

\begin{lemma}\label{lem:leq2}
  Let $\mu$ be a probability distribution on $\mathbb{RP}^1$ with
  finite support such that $\int_{x\in\mathbb{RP}^1}x\ \mathrm{d}\mu=\rho$,
  and let $H\in\cR\cP_>(\mathbb{C}^2)$.  Then, the set
\begin{eqnarray}
\label{region1}
\left\{M\in\B(\mu)\mid\left<M,H\right>=(1/2)\left<\rho,H\right>\right\}
\end{eqnarray}
is enclosed by a curve of length $\leq 2$.
\end{lemma}
This is proven in Appendix~\ref{app:leq2}.

\begin{proof}[Proof of Proposition~\ref{prop:circum}]
  Let $(\mu,\{f_\theta\})$ be a keyring local hidden state model for
  the steering ellipse of $\rho_{AB}$.  Let $H$ be a positive definite
  operator that is normal to the steering ellipse of $\rho_{AB}$ (such
  an operator exists because the tilt of the steering ellipse of
  $\rho_{AB}$ is less than 1 by assumption).  Because it is normal to
  the ellipse, $\langle\tilde{\rho}_B(\theta),H\rangle=u$ (independent
  of $\theta$). Choose a sequence $\mu_1, \mu_2, \ldots$ of
  probability distributions on $\mathbb{RP}^1$ with finite support
  which converges to $\mu$. Then, due to Lemma~\ref{lem:leq2} the sets
\begin{eqnarray}
\left\{M\in\B(\mu_i)\mid\left<M,H\right>=(1/2)\left<\rho_B,H\right>\right\},
\end{eqnarray}
are each enclosed by some curve of circumference $\leq 2$.  They
furthermore converge to the set
\begin{eqnarray}
\left\{M\in\B(\mu)\mid\left<M,H\right>=(1/2)\left<\rho_B,H\right>\right\}.
\end{eqnarray}
Because
$\langle\tilde{\rho}_B,H\rangle=\langle\tilde{\rho}_B(\theta),H\rangle+\langle\tilde{\rho}_B(\theta+\pi),H\rangle=2u$,
this set contains $\tilde{\rho}_B(\theta)$.  The desired result
follows.
\end{proof}

\section{The steering operator}

Proposition~\ref{prop:circum} gives a criterion for steerability that
is sufficient but not necessary.  In order to develop a criterion that
is both necessary and sufficient, we will need to work not with the
circumference of the steering ellipse, but with the following operator
whose trace is equal to the circumference of the steering ellipse:
\begin{eqnarray}
\label{thecentralop}
\int_0^{2\pi}\left|\deriv{\theta}\tilde{\rho}_B(\theta)\right|\mathrm{d}\theta.
\end{eqnarray}
It is easiest to work with cases in which (\ref{thecentralop}) is a
scalar multiple of $\rho_B$.  Our goal in the current section is to
show that for any $2$-qubit state $\rho_{AB}$ whose steering ellipse
has tilt $<1$, there is a $Y\in\cR\cD_>(\mathbb{C}^2)$ such that the
operator~\eqref{thecentralop} for
$\rho'_{AB}=(\id_A\ot Y)\rho_{AB}(\id_A\ot Y)$ is a scalar multiple of
$\rho'_B$.  This will enable the proof of our main result in
Section~\ref{sec:mainresult}.

The first two subsections will contain technical preparations.  First
we prove a concentration result for a particular type of integral.

\subsection{Integrals of the form $\int[F(x)/\sqrt{G(x)}]\ \mathrm{d}x$}

\begin{proposition}
\label{prop:gform}
Let $U\subset\mathbb{R}^n$ contain the origin in its interior, $F\colon U\to\mathcal{RP}_\geq(\mathbb{C}^2)$ be a continuous function
such that $F({\bf 0})\neq{\bf 0}$ and let
$G:U\to\mathbb{R}_{\geq 0}$ be twice differentiable with
$G({\bf x})=0$ if and only if ${\bf x}={\bf 0}$.
Then,
\begin{eqnarray}
\label{sqrtquant}
\lim_{(x_2,\ldots,x_n)\to{\bf 0}}\left<\int_{-a}^a\frac{F({\bf
  x})}{\sqrt{G({\bf x})}}\ \mathrm{d}x_1\right>&=&\left<F({\bf 0})\right>.
\label{qtarget}
\end{eqnarray}
\end{proposition}

\begin{proof}
  Since $G$ is twice differentiable and ${\bf x}={\bf 0}$ is a minimum
  of $G$, we have $\left|G({\bf x})\right|\leq C|{\bf x}|^2$ for some
  constant $C>0$.  Thus,
\begin{eqnarray}
\int_{-a}^a\frac{\mathrm{d}x_1}{G({\bf
    x})}&\geq&\frac{1}{C}\int_{-a}^a\frac{\mathrm{d}x_1}{\sqrt{x_1^2+y^2}}\\
&=&\frac{1}{C}\log\left(\frac{\sqrt{a^2+y^2}+a}{\sqrt{a^2+y^2}-a}\right)\\
&=&\frac{1}{C}\left(\log(1/y^2)+2\log(\sqrt{a^2+y^2}+a)\right)\, ,
\end{eqnarray}
where $y^2=\sum_{i=2}^nx_i^2$.  Since $y\to 0$ as
$(x_2,\ldots,x_n)\to(0,\ldots,0)$, this tends to $\infty$.
On the other hand, for any $\delta \in (0, a)$, 
\begin{eqnarray}
\lim_{(x_2,\ldots,x_n)\to{\bf 0}}\int_{[-a,
  a]\smallsetminus(-\delta,\delta)}\frac{\mathrm{d}x_1}{\sqrt{G({\bf
  x})}}=\int_{[-a,a]\smallsetminus(-\delta,\delta)}\frac{\mathrm{d}x_1}{\sqrt{G(x_1,0,\ldots,0)}}<\infty,
\end{eqnarray}
since we assumed that $G({\bf x})$ has only one zero.  Thus
as $(x_2,\ldots,x_n)\to{\bf 0}$, the integral of
$F(\mathbf{x})/\sqrt{G({\bf x})}$ on
$(-\delta,\delta)$ dominates the integral of the same quantity
on $[-a,a]\smallsetminus(-\delta,\delta)$.  The quantity on the left
side of Equation~\eqref{sqrtquant} is therefore in the convex
hull of $F((-\delta,\delta))$.  Since this holds true
for any $\delta>0$, Equation~\eqref{sqrtquant} follows.
\end{proof}

\subsection{Formulas for the absolute value of a $2 \times 2$ matrix}

Throughout this section, $X$ and $Y$ denote $2\times 2$ real symmetric
matrices.  For any such matrix
$Y=\left[\begin{array}{cc} d & e \\ e & f \end{array}\right]$, let
$\widehat{Y}=\left[\begin{array}{cc} f & -e \\ -e &
    d \end{array}\right]$
denote the adjugate matrix.  (The adjugate matrix has the same
eigenspaces as $Y$, with the two eigenvalues interchanged.)  Note that
$Y\widehat{Y}=\det(Y)\id$.

\begin{definition}
If $X,Y\in\mathcal{RA}(\mathbb{C}^2)$ and $Y$ is invertible, let
\begin{eqnarray}
\left|X\right|_Y=Y^{-1}\left|YXY\right|Y^{-1}.
\end{eqnarray}
\end{definition}

Note that if $X$ is positive semidefinite, then its trace-norm and absolute value are easily computed:
$\|X\|_1=\Tr(X)$, and $|X|=X$.  The next propositions compute these values
in the case where $X$ is neither positive semidefinite nor negative semidefinite.

\begin{proposition}
If $X$ is such that $X\ngeq 0$ and $X\nleq 0$, then
\begin{eqnarray}
\label{traceeq1}
\left\|X\right\|_1&=&\sqrt{\Tr(X^2-X\widehat{X})} \\
\label{abseq}
\left|X\right|&=&\frac{X^2-X\widehat{X}}{\left\|X\right\|_1}.
\end{eqnarray}
\label{prop:sf}
\end{proposition}

\begin{proof}
Direct computation.
\end{proof}

\begin{proposition}
\label{prop:yxcomp}
If $X$ is such that $X\ngeq 0$ and $X\nleq 0$, and $Y$ is invertible, then
\begin{eqnarray}
\left|X\right|_Y&=&\frac{XY^2X-(\det{X})\widehat{Y}^2}{\left\|YXY\right\|_1}.
\end{eqnarray}
\end{proposition}

\begin{proof}
See Appendix~\ref{app:yxcomp}.
\end{proof}

Note that $\left\|YXY\right\|_1^2$ is a polynomial in the entries
of $X$ and $Y$ (via (\ref{traceeq1})) and is therefore infinitely
differentiable as a function of $X$ and $Y$.

\subsection{The steering operator of a two-qubit state}\label{app:C}
We now apply the results from the previous subsections.  Suppose, that
that $\rho_{AB}$ is a two-qubit state and that its steering ellipse
$\{\tilde{\rho}_B(\theta)\mid\theta\in\mathbb{R}\}$ has tilt less than
$1$.  Define function
$X\colon\mathbb{R}\to\mathcal{RA}(\mathbb{C}^2)$ so that
\begin{eqnarray}
\label{abcdef}
X(\theta)=\deriv{\theta}\tilde{\rho}_B(\theta)=\Tr_A((D(\theta)\ot\id)\rho_{AB})\, ,
\end{eqnarray}
where
$D(\theta)=\frac{1}{2}(-\sin\theta\proj{0}+\cos\theta(\ketbra{0}{1}+\ketbra{1}{0})+\sin\theta\proj{1})$.

Note that because the steering ellipse
$\{\tilde{\rho}_B(\theta)\}$ has tilt $<1$, for every $\theta$
the operator $X(\theta)$ is neither positive semidefinite nor
negative semidefinite (cf.\ Corollary~\ref{cor:app1}).

Let $P\in\mathbb{RP}^1$.  Let
$Y\colon\mathbb{R}^2\to\mathcal{RP}_{\geq}(\mathbb{C}^2)$ be given by
\begin{eqnarray}
Y(r_1,r_3)&=&P+r_1\sigma_1+r_3\sigma_3.
\end{eqnarray}
The function $\theta\mapsto\Tr(PX(\theta))$ varies sinusoidally
and has exactly two zeros in $[0,2\pi)$. Without loss of generality,
we will assume that the zeros are $\theta=0$ and $\theta=\pi$.  We
wish to compute
\begin{eqnarray}
&& \lim_{(r_1,r_3) \to (0,0)} \left< \int_{-\pi/2}^{\pi/2} \left| X \right|_{Y} \mathrm{d} \theta \right> \\
\label{ultexp}
& = & \lim_{(r_1,r_3) \to (0,0)} \left< \int_{-\pi/2}^{\pi/2} \frac{
      XY^2 X - (\det X ) \widehat{Y}^2 }{\sqrt{ \left\| Y X Y
      \right\|_1^2}}\ \mathrm{d} \theta \right>.
\end{eqnarray}
The function $\left\| P X(\theta) P \right\|^2_1 
= (\Tr ( P X(\theta) ))^2$  on the interval $[-\pi/2 , \pi/2]$ has a zero only at $\theta = 0$.  By Proposition~\ref{prop:gform}
(with $G(\theta, r_1, r_3 ) = \left\| Y X Y \right\|_1^2$
and $F( \theta, r_1, r_3 )$ equal to the numerator of the integrand in (\ref{ultexp})), we obtain the following:
\begin{eqnarray}
\lim_{(r_1,r_3) \to (0,0)} \left< \int_{-\pi/2}^{\pi/2} \left| X  \right|_{Y} \mathrm{d} \theta \right> & = & \left< X ( 0) P^2 X ( 0 ) - \det (X ( 0 ) )
\widehat{P}^2 \right> \\
& = & \left<-2 \det (X ( 0 ) ) \widehat{P}^2 \right> \\
& = & \widehat{P}.
\end{eqnarray}
Exploiting symmetry, the same equality holds when we replace
the upper and lower integral limits with $-\pi/2$ and $3\pi/2$ (or equivalently, with $0$ and $2 \pi$).  We therefore have the following.
\begin{theorem}
\label{theorem:boundary}
Let $\rho_{AB}$ be a two-qubit state whose steering ellipse $\{ \tilde{\rho}_B ( \theta) \mid \theta \in 
\mathbb{R} \}$ has tilt less than $1$.  Then, for any $P \in \mathbb{RP}^1$,
\begin{eqnarray}
\lim_{Y \to P} \left< \int_{0}^{2 \pi} \left| \deriv{\theta} \tilde{\rho}_B ( \theta )  \right|_{Y} \mathrm{d} \theta \right>  =  \widehat{P}=\id-P,
\end{eqnarray}
where the limit is taken over all positive definite density operators
$Y$.
\end{theorem}

As a consequence of Theorem~\ref{theorem:boundary}, the function
$\mathcal{RD}_> ( \mathbb{C}^2 ) \to \mathcal{RD} ( \mathbb{C}^2 )$
given by
\begin{eqnarray}
\label{normmap}
Y \mapsto \left< \int_0^{2 \pi} \left| \deriv{\theta} \tilde{\rho}_B ( \theta ) \right|_{Y} \mathrm{d} \theta \right>
\end{eqnarray}
extends continuously to a map $\mathcal{RD}(\mathbb{C}^2)\to
\mathcal{RD}(\mathbb{C}^2)$ which has the effect of mapping each
element of $\mathbb{RP}^1$ to its orthogonal complement [see,
  for example,
  Theorem~D on Page~78 of~\cite{Simmons}.].  By
Lemma~\ref{toplemma} in the appendix, the function given
by~\eqref{normmap} is onto.  In particular, its image contains
$\rho_B$.  We therefore have the following.

\begin{lemma}
\label{ontolemma}
Let $\rho_{AB}\in\cR\cD(\mathbb{C}^2\ot\mathbb{C}^2)$ be a
two-qubit state whose steering ellipse has tilt $<1$.  Then, there
exists $Y\in\mathcal{RD}_>(\mathbb{C}^2)$ such that
\begin{eqnarray}
\int_0^{2 \pi}\left|\deriv{\theta}\tilde{\rho}_B(\theta)\right|_{Y}\mathrm{d}\theta
\end{eqnarray}
is a scalar multiple of $\rho_B$.
\end{lemma}

Note that if we use
$\rho'_{AB}=(\id_A\ot Y)\rho_{AB}(\id_A\ot Y)$ in
Lemma~\ref{ontolemma} then we have that
\begin{eqnarray}
\int_{0}^{2\pi}\left|\deriv{\theta}\tilde{\rho}'_B(\theta)\right|\mathrm{d}\theta
\end{eqnarray}
is a scalar multiple of $\rho'_B$, which was our original goal.

\section{A criterion for RP-steerability}

\label{sec:mainresult}

Now we are ready to prove a criterion for RP-steerability that is both
necessary and sufficient.  The next theorem and corollary contain our
main result.

\begin{theorem}
\label{steerthm2}
Let $\rho_{AB}\in\cR\cD(\mathbb{C}^2\ot\mathbb{C}^2)$ be a two-qubit state whose steering ellipse
has tilt $<1$.  Then, $\rho_{AB}$ is RP-unsteerable if and only if
there exists $Y\in\cR\cP_>(\mathbb{C}^2)$ such that
\begin{eqnarray}
\label{steerbound2}
Y\rho_BY-\int_0^{\pi} \left|Y\deriv{\theta} \left( \tilde{\rho}_B \left(
                \theta \right) \right) Y\right|\ \mathrm{d} \theta\geq 0\, .
\end{eqnarray}
\end{theorem}
\begin{corollary}\label{cor:main}
  Let $\rho_{AB}\in\cR\cD(\mathbb{C}^2\ot\mathbb{C}^2)$ be a two-qubit
  state whose steering ellipse has tilt $<1$. Then $\rho_{AB}$ is
  RP-steerable if and only if there exists
  $Y\in\cR\cP_>(\mathbb{C}^2)$ such that
\begin{eqnarray}
\label{steerbound_conv}
Y\rho_BY-\int_0^{\pi}\left|Y\deriv{\theta}\left(\tilde{\rho}_B\left(\theta\right)\right)Y\right|\mathrm{d}\theta\leq 0\, ,
\end{eqnarray}
with the left-hand-side not equal to 0.
\end{corollary}

Note that~\eqref{steerbound2} can be 
rewritten as 
\begin{eqnarray}
\label{steerbound}
\rho_B-\int_0^{\pi} \left| \deriv{\theta} \left( \tilde{\rho}_B \left(
                \theta \right) \right) \right|_Y \mathrm{d} \theta\geq 0\, .
\end{eqnarray}

The following result found
in~\cite{Quintino15} will be important for the proofs that follow.

\begin{lemma}\label{lem:M}
  If $\rho_{AB}$ has a LHS model (for any set of measurements), then so does
  $\langle(\cI\ot\cM)(\rho_{AB})\rangle$
  for any positive linear map $\cM$.
\end{lemma}

In particular, for any invertible Hermitian operator $Y$, $\rho_{AB}$
is RP-steerable if and only if
$\left\langle(\id_A\ot Y)\rho_{AB}(\id_A\ot Y)\right\rangle$ is RP-steerable.

\begin{proof}[Proof of Theorem~\ref{steerthm2}.]
  For any Hermitian operator $X$, define $|X|_\pm:=(|X|\pm X)/2$, and
  $\left\|X\right\|_{\pm} = \Tr\left|X\right|_{\pm}$.

\textbf{Case 1:} Suppose
\begin{eqnarray}
\label{circumineq}
\rho_B & \geq & \rho':=
 \int_0^{\pi} \left| \deriv{\theta} \left( \tilde{\rho}_B \left( \theta
                \right) \right) \right| \mathrm{d} \theta\, ,
 \end{eqnarray}
and define
\begin{eqnarray}
\sigma_\lambda & := & \left| \deriv{\lambda} \left( \tilde{\rho}_B \left( \lambda \right) \right) \right|_+
+ \frac{\rho_B - \rho'}{2\pi}.
\end{eqnarray}
Because $\tilde{\rho}_B ( \lambda + \pi ) = \rho_B - \tilde{\rho}_B ( \lambda)$,
the operator
$(\mathrm{d} / \mathrm{d} \lambda ) \tilde{\rho}_B ( \lambda + \pi ) $ is the negation of the
operator $ (\mathrm{d} / \mathrm{d} \lambda ) \tilde{\rho}_B ( \lambda )$, and so the
following equality also holds:
\begin{eqnarray}
\sigma_\lambda & = & \left| \deriv{\lambda} \left( \tilde{\rho}_B \left( \lambda + \pi \right) \right) \right|_-
+ \frac{\rho_B - \rho'}{2\pi}.
\end{eqnarray}

We proceed to construct a local hidden state model from
$\{ \sigma_\lambda \}_\lambda$.  We have the following:
\begin{eqnarray}
\int_0^{2 \pi}  \sigma_\lambda\ \mathrm{d} \lambda & = & \int_0^{2 \pi} \left| 
\deriv{\lambda} \tilde{\rho}_B \left( \lambda \right)  \right|_+ \mathrm{d} \lambda +
\rho_B - \rho'   \\
& = & \int_0^{\pi} \left| 
\deriv{\lambda} \tilde{\rho}_B \left( \lambda \right)  \right|_+ \mathrm{d} \lambda +
\int_{\pi}^{2 \pi} \left| 
\deriv{\lambda} \tilde{\rho}_B \left( \lambda \right)  \right|_+ \mathrm{d} \lambda +
( \rho_B - \rho' )  \\
& = & \int_0^{\pi} \left| 
\deriv{\lambda} \tilde{\rho}_B \left( \lambda \right)  \right|_+ \mathrm{d} \lambda +
\int_0^{\pi} \left| 
\deriv{\lambda} \tilde{\rho}_B \left( \lambda \right)  \right|_- \mathrm{d} \lambda +
( \rho_B - \rho' )  \\
& = & \int_0^{\pi} \left| 
\deriv{\lambda} \tilde{\rho}_B \left( \lambda \right)  \right| \mathrm{d} \lambda +
( \rho_B - \rho' )  \\
& = & \rho' + \rho_B - \rho' \\
& = & \rho_B.
\end{eqnarray}
For any $\theta \in [0, \pi]$ let $g_\theta \colon \mathbb{RP}^1 \to [0, 1]$ be equal to zero
on the interval $[\theta, \theta + \pi]$ and equal to $1$ elsewhere, and define
$g_\theta$ for $\theta \in (\pi, 2 \pi ]$ by $g_{\theta} = 1 - g_{\theta - \pi}$.
Then,
\begin{eqnarray*}
\int_0^{2 \pi}\!\!\! g_\theta ( \lambda) \sigma_\lambda\ \mathrm{d} \lambda  & = &
\frac{1}{2} \left[ \int_0^{2 \pi} ( 2 g_\theta ( \lambda) - 1) \sigma_\lambda\ \mathrm{d} \lambda
+ \int_0^{2 \pi } \sigma_\lambda\ \mathrm{d} \lambda \right] \\
& = &
\frac{1}{2} \left[-\!\int_\theta^{\theta + \pi \textbf{ mod } 2 \pi }  \sigma_\lambda\ \mathrm{d} \lambda
+ \int_{\theta+ \pi \textbf{  mod } 2 \pi}^{\theta + 2 \pi \textbf{ mod } 2 \pi }  \sigma_\lambda\ \mathrm{d} \lambda
+ \rho_B \right] \\
& = &
\frac{1}{2} \left[-\!\int_\theta^{\theta + \pi \textbf{ mod } 2 \pi }\!\! \left(\left| \deriv{\lambda} \tilde{\rho}_B ( \lambda ) \right|_+\!\!\! - \left| \deriv{\lambda} \tilde{\rho}_B ( \lambda ) \right|_-\right) \mathrm{d} \lambda\!
+\! \rho_B \right] \\
& = &
\frac{1}{2} \left[-\!\int_\theta^{\theta + \pi \textbf{ mod } 2 \pi } \deriv{\lambda} \tilde{\rho}_B ( \lambda )\ \mathrm{d} \lambda
+ \rho_B \right] \\
& = &
\frac{1}{2} \left[ - \tilde{\rho}_B ( \theta + \pi ) + \tilde{\rho}_B ( \theta ) 
+ \rho_B \right] = \tilde{\rho}_B ( \theta )\, .
\end{eqnarray*}
Thus $\{ \tilde{\rho}_B ( \theta ) \}_\theta$ has a local hidden state model.

\textbf{Case 2:} Suppose that there exists $Y\in\cR\cP_>(\mathbb{C}^2)$
such that
\begin{eqnarray}
Y\rho_BY & \geq & 
 \int_0^{\pi} \left|Y \deriv{\theta} \left( \tilde{\rho}_B \left( \theta
                \right) \right) Y\right| \mathrm{d}\theta\, .
 \end{eqnarray}

In this case, the state
\begin{eqnarray}
\overline{\rho_{AB}} = \langle(\id\ot Y) \rho_{AB} (\id \ot Y)\rangle 
\end{eqnarray}
satisfies the conditions of Case 1.  Since
$\cM:X\mapsto Y^{-1}XY^{-1}$ is a positive map, by Lemma~\ref{lem:M}, a
local hidden state model exists for $\rho_{AB}$.

\textbf{Case 3:} Suppose that for all $Y\in\cR\cP_>(\mathbb{C})$, 
\begin{eqnarray}
Y\rho_BY & \ngeq & I_Y:=
 \int_0^{\pi} \left|Y \deriv{\theta} \left( \tilde{\rho}_B \left( \theta \right) \right)Y \right| \mathrm{d}\theta
\end{eqnarray}
By Lemma~\ref{ontolemma}, we can find $Y$ such that $I_Y$ is a scalar
multiple of $Y\rho_B Y$ (this is why Corollary~\ref{cor:main} follows
from Theorem~\ref{steerthm2}).  Thus we have
\begin{eqnarray}
Y\rho_B Y&=& c  \int_0^{\pi} \left| Y\deriv{\theta}  \left( \rho_B \left( \theta \right) \right) Y\right| \mathrm{d}\theta
\end{eqnarray}
for some $c < 1$.    Letting $\gamma_{AB}=\langle(\id \ot Y) \rho_{AB} (\id\ot Y)\rangle$,
we have
\begin{eqnarray}
\gamma_B & = & 
c  \int_0^{\pi} \left| \deriv{\theta}  \left(  \gamma_B \left( \theta \right) \right) \right| \mathrm{d}\theta
\end{eqnarray}
which in particular means
\begin{eqnarray}
\int_0^{\pi} \left\| \deriv{\theta}  \left( \gamma_B \left( \theta \right) \right) \right\|_1 \mathrm{d}\theta & \geq &
(1/c) \Tr ( \gamma_B ) > 1.
\end{eqnarray}
By symmetry, replacing the upper limit ($\pi$) in the integral above has the effect of doubling its value; thus,
\begin{eqnarray}
\label{twopi}
\int_0^{2 \pi} \left\| \deriv{\theta}  \left(  \gamma_B \left( \theta \right) \right) \right\|_1 \mathrm{d}\theta & > &  2,
\end{eqnarray}
which implies by Proposition~\ref{prop:circum} that $\gamma$ (and therefore $\rho$) has no local hidden variable model.
\end{proof}

\section{Explicit calculations for steering ellipses}

\subsection{Application I: RP-steerability of Werner states}
It is interesting to see what this criteria gives for Werner states,
i.e., the family
$\rho_{AB}(\eta)=\eta\proj{\Phi_+}+(1-\eta)\id/4$ where $\eta\in[0,1]$
and $\ket{\Phi_+}=\frac{1}{\sqrt{2}}(\ket{00}+\ket{11})$.

\begin{theorem}\label{thm:Werner}
  States of the form $\rho_{AB}(\eta)$ are RP-unsteerable for
  $\eta\leq\frac{2}{\pi}$ and are RP-steerable for
  $\eta>\frac{2}{\pi}$.
\end{theorem}
\begin{proof}
The steering
ellipses for these states are
$\tilde{\rho}_B(\theta)=\frac{1}{4}\left(
\begin{array}{cc}
 1+\eta\cos \theta & \eta  \sin \theta \\
 \eta  \sin \theta & 1-\eta  \cos \theta
\end{array}
\right)$
and have zero tilt for all $\eta$ (since all these states have the
same trace, the difference between any two states on the ellipse is
orthogonal to $\id/2$). The derivative with respect to $\theta$ is
$\deriv{\theta}(\tilde{\rho}_B(\theta))=\frac{\eta}{4}\left(
\begin{array}{cc}
 -\sin \theta & \cos \theta \\
 \cos \theta & \sin \theta \\
\end{array}
\right)$
which has
$|\deriv{\theta}(\tilde{\rho}_B(\theta))|=\frac{\eta}{4}\id$.  Hence,
$$\rho_B-\int_0^{\pi} \left|\deriv{\theta} \left( \tilde{\rho}_B \left(
      \theta \right) \right) \right| \mathrm{d}
\theta=\id/2-\frac{\pi\eta}{4}\id\, .$$
Applying Theorem~\ref{steerthm2} and Corollary~\ref{cor:main} with
$Y=\id$ we have that Werner states are RP-unsteerable if
$\frac{\pi\eta}{4}\leq\frac{1}{2}$, i.e.,
$\eta\leq\frac{2}{\pi}\approx 0.637$ and are RP-steerable if
$\eta>\frac{2}{\pi}$.
\end{proof}
Note that this boundary was already known~\cite{Jones11,Uola16}, and
that it is possible to get close to this bound with small numbers of
measurements~\cite{Jones11,Bavaresco}.

\subsection{Application II: RP-steerability of partially entangled
  states mixed with uniform noise}
Consider the family $\rho_{AB}(\alpha,\eta):=\eta\proj{\phi_\alpha}+(1-\eta)\id/4$, where
$\ket{\phi_\alpha}:=\cos\alpha\ket{00}+\sin\alpha\ket{11}$ for
$0\leq\alpha\leq\frac{\pi}{4}$.  The
steering ellipses for these states are $\tilde{\rho}_B^{\,\alpha,\eta}(\theta)=\left(
\begin{array}{cc}
 \eta  \cos ^2(\alpha) \cos ^2\left(\frac{\theta}{2}\right)-\frac{\eta
   }{4}+\frac{1}{4} & \frac{1}{2} \eta  \cos (\alpha) \sin (\alpha)
   \sin (\theta) \\
 \frac{1}{2} \eta  \cos (\alpha) \sin (\alpha) \sin (\theta) & \eta  \sin
   ^2(\alpha) \sin ^2\left(\frac{\theta}{2}\right)-\frac{\eta }{4}+\frac{1}{4}
\end{array}
\right)$
and are plotted in the Bloch representation in
Fig.~\ref{fig:ellipses}.

\begin{figure}
\includegraphics[width=0.4\textwidth]{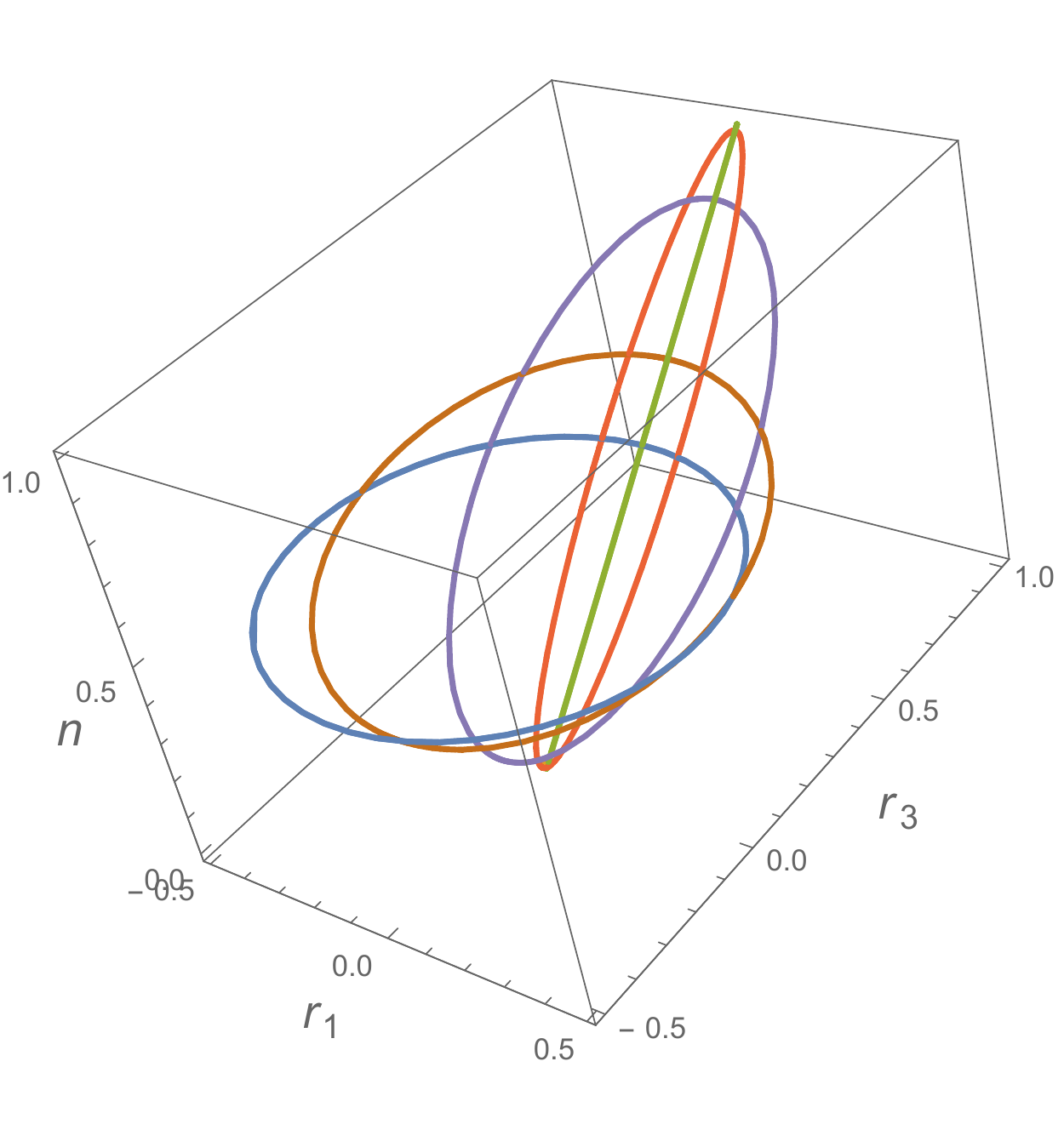} \includegraphics[width=0.4\textwidth]{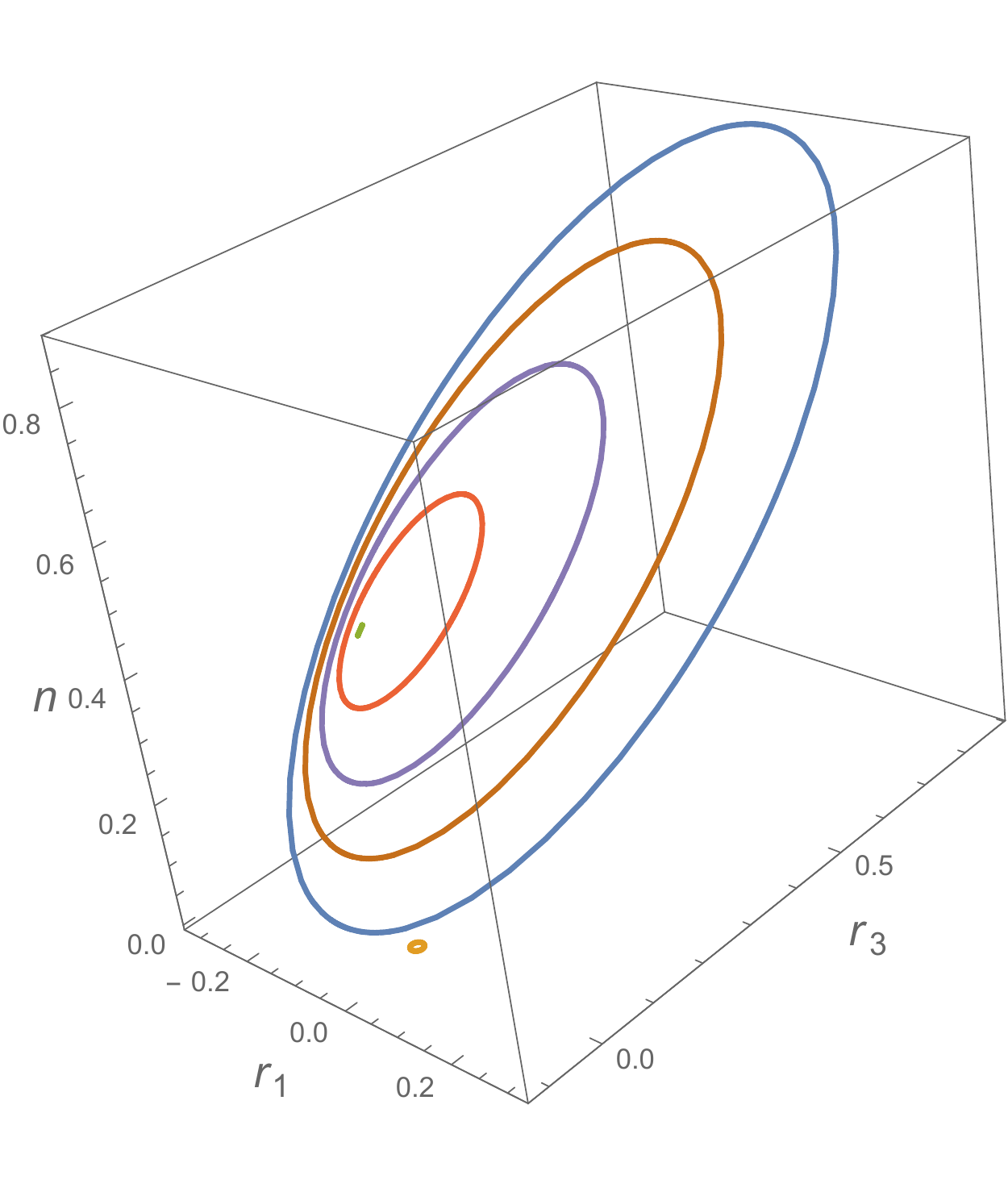} 
\center{(a)\qquad\qquad\qquad\qquad\qquad\qquad\qquad\qquad\qquad\qquad\qquad (b)}
\caption{(a) Steering ellipses in the Bloch representation for
  $\eta=1$, $\alpha=\pi/4$ (blue), $0.65$ (brown), $0.35$ (purple),
  $0.1$ (red) and $0$ (green); (b) $\alpha=0.35$ and $\eta=1$
  (blue), $\frac{3}{4}$ (brown), $\frac{1}{2}$ (purple), $\frac{1}{4}$
  (red) and $\eta=0.01$ (green).  The small yellow circle on the right
  marks the origin.}
\label{fig:ellipses}
\end{figure}

One can verify that for
$A_\alpha=\left(\begin{array}{cc}\sin^2\alpha&0\\0&\cos^2\alpha\end{array}
\right)$,
$\Tr(A_{\alpha}\tilde{\rho}_B^{\,\alpha,\eta}(\theta))=\frac{1}{8}(2-\eta(1-\cos(4\alpha)))$,
which is independent of $\theta$.  $A_\alpha$ is hence normal to the
steering ellipse and so the tilt of the ellipse is
$\cos(2\alpha)\leq 1$, and approaches~$1$ as $\alpha$ approaches~$0$.
\begin{remark}
  The tilt is independent of $\eta$ and hence the steering ellipse for
  any two-qubit pure state has tilt at most~1.
\end{remark}

We have
$\rho_B(\alpha,\eta)=\frac{1}{2}(1+\eta\cos(2\alpha))\proj{0}+\frac{1}{2}(1-\eta\cos(2\alpha))\proj{1}$.

The derivative of the steering ellipse with respect to $\theta$ is
\begin{equation}\label{deriv:1}
\deriv{\theta}\tilde{\rho}_B^{\,\alpha,\eta}(\theta)=\frac{\eta}{2}\left(
\begin{array}{cc}
 -\cos ^2(\alpha) \sin\left(\theta\right)& \cos (\alpha) \sin (\alpha)   \cos (\theta) \\
   \cos (\alpha) \sin (\alpha) \cos (\theta) & \sin^2(\alpha) \sin\left(\theta\right)
\end{array}
\right)\, .
\end{equation}

For $\alpha=\frac{\pi}{4}$ the case is as before.  To investigate
other values of $\alpha$, we note that, by Remark~\ref{rmk:convex}, if
$\rho_{AB}(\alpha,\eta)$ has a LHS model, then so does
$\rho_{AB}(\alpha,\eta')$ for $\eta'<\eta$.  Thus, for each $\alpha$
there is a critical value $\bar{\eta}(\alpha)$ such that
$\rho_{AB}(\alpha,\eta)$ is RP-steerable for $\eta>\bar{\eta}(\alpha)$
and is RP-unsteerable for $\eta\leq\bar{\eta}(\alpha)$.  We search for
this critical value numerically.  

Since $Y$ has real entries, is positive and multiplying by a constant
doesn't affect whether~\eqref{steerbound2} holds, we can take $Y$ to
have $\Tr(Y)=1$ and parameterize it in terms of two parameters $r_1$
and $r_3$ using a plane of the Bloch sphere via
$Y=\frac{1}{2}(\id+r_1\sigma_1+r_3\sigma_3)$. To do the search we use
the following subroutines:
\begin{enumerate}
\item\label{sub1} For fixed $\alpha$ and $\eta$ this searches over
  $r_1,r_3$ to find the largest value of the minimum eigenvalue of the
  expression on the left of~\eqref{steerbound2}. This uses gradient
  ascent with decreasing step-size, terminating when no improvement
  can be found for some minimal step-size, or when $r_1,r_3$ are found
  such that the minimum eigenvalue is positive
  (i.e.,~\eqref{steerbound2} is satisfied).  The output is either the
  largest value found or the first positive value found.
\item\label{sub2} This is analogous to Subroutine~\ref{sub1}, except
  it searches for the smallest value of the maximum eigenvalue of the
  expression on the left of~\eqref{steerbound2}, terminating either
  when a negative value is obtained or when no improvement can be found
  for some minimal step-size.
\item\label{sub3} For fixed $\alpha$, this uses Binary Search to find
  the largest $\eta$ for which Subroutine~\ref{sub1} returns a
  positive value, for some number of search steps.
\item\label{sub4} For fixed $\alpha$, this uses Binary Search to find
  the smallest $\eta$ for which Subroutine~\ref{sub2} returns a
  negative value, for some number of search steps.
\end{enumerate}

Subroutine~\ref{sub3} hence gives a certified lower bound on
$\bar{\eta}(\alpha)$ and Subroutine~\ref{sub4} a certified upper
bound.  By varying the step-sizes and number of steps, in principle,
we can make the gap between these as small as we like (in practice,
the limits of machine precision provide a cut-off).

Note that if Subroutine~\ref{sub1} has a negative output, we cannot
strictly rule out that there exists a $Y$ such that
condition~\eqref{steerbound2} holds: in principle a smaller
step-size might reveal a suitable $Y$.  This is why we use
Subroutine~\ref{sub2} in parallel.

The result is given in Figure~\ref{introfig} (although the plot only
shows $\eta>0.6$, the region extends to $\eta=0$).

\subsection{RP-steerability of depolarizing channel states}
Consider a source that generates an entangled state that is sent to
two parties via two depolarizing channels with parameters $\eta_A$ and
$\eta_B$,
i.e., these channels take
$$\cS(\mathbb{C}^2\ot\mathbb{C}^2)\to\cS(\mathbb{C}^2\ot\mathbb{C}^2):\rho_{AB}\mapsto\hat{\rho}_{AB}:=(\cE_{\eta_A}\ot\cE_{\eta_B})(\rho_{AB})\,,$$
where $\cE_\eta:\cS(\mathbb{C}^2)\to\cS(\mathbb{C}^2)$ is given by $\cE_\eta(\rho)=\eta\rho+(1-\eta)\id/2$.

For $\rho_{AB}=\proj{\Phi_+}$, this channel leads to Werner states
(with parameter $\eta_A\eta_B$ instead of $\eta$).  The states
are hence RP-unsteerable iff $\eta_A\eta_B\leq\frac{2}{\pi}\approx0.637$.

More generally, for $\rho_{AB}=\proj{\phi_\alpha}$, we call the state
after the channel $\hat{\rho}_{AB}(\alpha,\eta_A,\eta_B)$ and note
that
$$\hat{\rho}_B=\frac{1}{2}\left((1+\eta_B\cos(2\alpha))\proj{0}+(1-\eta_B\cos(2\alpha))\proj{1}\right)$$
is independent of $\eta_A$.  The steering ellipse for such a state is
$$\tilde{\rho}_B^{\,\alpha,\eta_A,\eta_B}(\theta)=\frac{1}{4}\left(
\begin{array}{cc}
  1+\eta_A\cos(2\alpha)\cos(\theta)+\eta_B(\eta_A\cos(\theta)+\cos(2\alpha))
    & \eta_A\eta_B\sin(2\alpha)\sin(\theta) \\
  \eta_A\eta_B\sin(2\alpha)\sin(\theta) & 1+\eta_A\cos(2\alpha)\cos(\theta)-\eta_B(\eta_A\cos(\theta)+\cos(2\alpha))
\end{array}
\right)$$

For $A_{\alpha,\eta_A,\eta_B}=\left(
\begin{array}{cc}\frac{\eta_B-\cos(2\alpha)}{2\eta_B}&0\\0&\frac{\eta_B+\cos(2\alpha)}{2\eta_B}\end{array}\right)$,
we have
$\Tr(A_{\alpha,\eta_A,\eta_B}\tilde{\rho}_B^{\,\alpha,\eta_A,\eta_B}(\theta))=\frac{1}{2}\sin^2(2\alpha)$,
which is independent of $\theta$.  Hence $A_{\alpha,\eta_A,\eta_B}$ is
the normal to the steering ellipse, and the ellipse has tilt
$\frac{\cos(2\alpha)}{\eta_B}$.  This is less than 1 for
$\eta_B>\cos(2\alpha)$, so we can use Theorem~\ref{steerthm2} and
Corollary~\ref{cor:main} provided this holds.

The derivative of the steering ellipse is 
\begin{equation}
\deriv{\theta}\tilde{\rho}_B^{\,\alpha,\eta_A,\eta_B}(\theta)=\frac{\eta_A}{4}\left(
\begin{array}{cc}
  -(\eta_B+\cos(2\alpha))\sin(\theta)
  &  
    \eta_B\sin(2\alpha)\cos(\theta) \\
  \eta_B\sin(2\alpha)\cos(\theta) & (\eta_B-\cos(2\alpha))\sin(\theta)
\end{array}
\right)\, .
\end{equation}

Since this is proportional to $\eta_A$, the amount of noise on Alice's
side (the untrusted side), the case of noise only on Bob's side is
representative of the general case.

We first make two observations for special cases, before proceeding
with the general case:
\begin{enumerate}
\item If there is no noise on Bob's side (i.e., the trusted side),
  i.e., if $\eta_B=1$, then
  $\deriv{\theta}\tilde{\rho}_B^{\,\alpha,\eta_A,1}(\theta)$ is
  identical to that in~\eqref{deriv:1}, and the tilt of the steering
  ellipse of $\rho_{AB}(\alpha,\eta_A,1)$ is $\cos(2\alpha)\leq 1$, so
  we obtain the same result.
\item If the state is maximally entangled, i.e.,
  $\alpha=\frac{\pi}{4}$, then the situation is exactly the same as
  for a Werner state with $\eta=\eta_A\eta_B$. In other words,
  $\eta_A\eta_B\leq\frac{2}{\pi}$ is a necessary and sufficient
  condition for RP-unsteerability of a state of the form
  $\rho_{AB}(\pi/4,\eta_A,\eta_B)$.
\end{enumerate}

We study the general case numerically, using similar techniques to
before.  The results are shown in Figure~\ref{fig:gen}.

\begin{figure}
\includegraphics[width=0.4\textwidth]{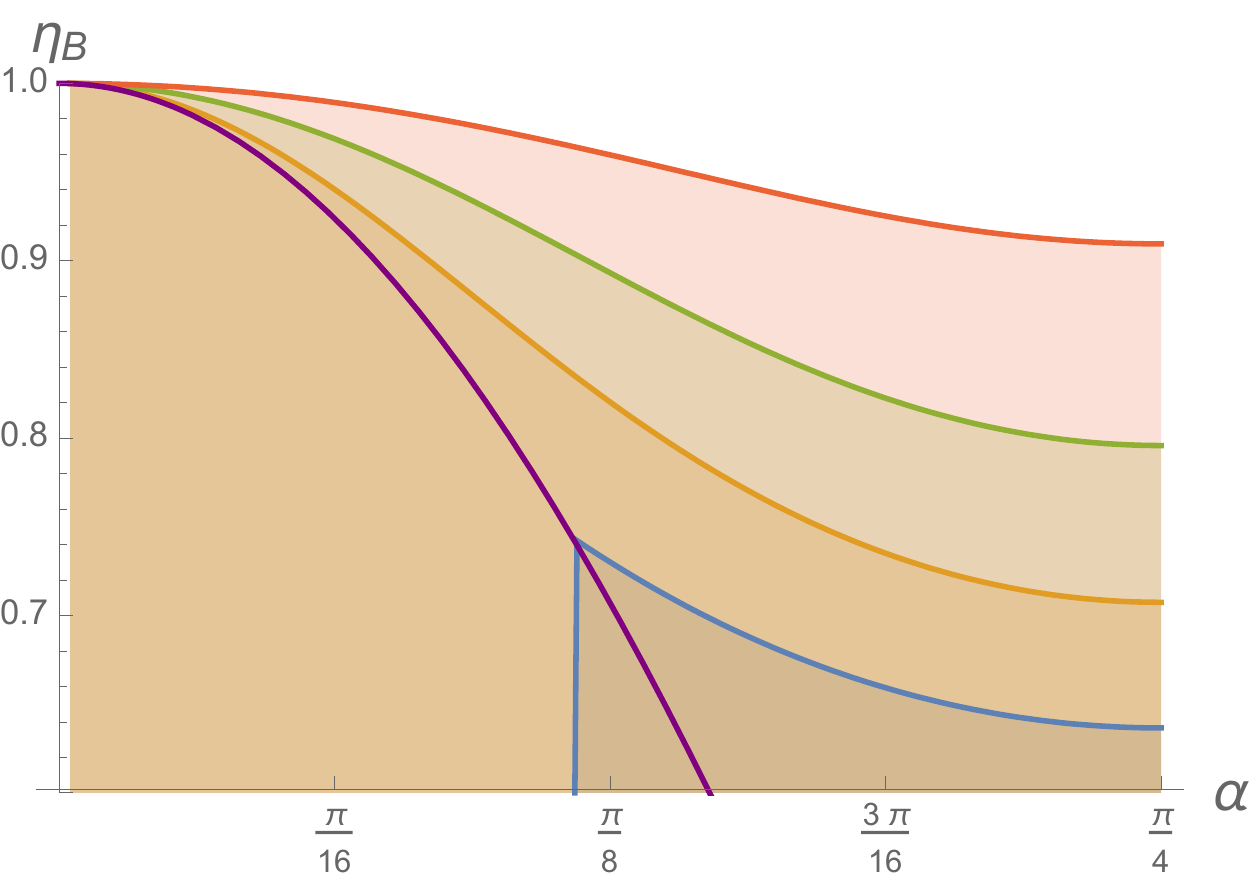}
\caption{Plot of the regions where a LHS model exists for all real
  projective measurements for $\eta_A=1$ (blue), $\eta_A=0.9$ (orange)
  and $\eta_A=0.8$ (green), $\eta_A=0.7$ (red) (although not shown,
  all regions extend downwards to $\eta_B=0$), together with the
  purple curve $\eta_B=\cos(2\alpha)$ which we need to be above to use
  Theorem~\ref{steerthm2} and Corollary~\ref{cor:main}.  In the case
  $\eta_A=\frac{2}{\pi}$ (not shown), the state is RP-unsteerable for
  all $\eta_B$ and $\alpha$. For $\eta_A=0.9$, $0.8$ and $0.7$ we have
  a complete classification: above each of the corresponding regions,
  the state is RP-steerable.  In the case $\eta_A=1$, the
  classification is incomplete for $\alpha\lessapprox0.37$. Here, if
  $\eta_B\leq\cos(2\alpha)$ we are unable to decide whether or not the
  states are RP-steerable (while for $\eta_B>\cos(2\alpha)$ we know
  the states are RP-steerable).}
\label{fig:gen}
\end{figure}

The left hand side of~\eqref{steerbound2} becomes easier to satisfy
for lower $\eta_A$ and so the region of RP-unsteerability increases as
$\eta_A$ is lowered.  In other words, if
$\hat{\rho}_{AB}(\alpha,\eta_A,\eta_B)$ is RP-unsteerable, then so is
$\hat{\rho}_{AB}(\alpha,\eta'_A,\eta_B)$ for $\eta'_A\leq\eta_A$.  At
$\eta_A=\frac{2}{\pi}$ the state is RP-unsteerable for all $\eta_B$
and $\alpha$.

Note that the regions shown in the above plot extend below the purple
curve, although the condition on the tilt of the steering ellipse
ceases to be satisfied there.  To extend to this region we use the
fact that more noise (lower $\eta_B$) makes a LHS model easier to
construct.  This is stated in the following lemma.
\begin{lemma}
  If $\hat{\rho}_{AB}(\alpha,\eta_A,\eta_B)$ has a LHS model (for any
  set of measurements), then so does
  $\hat{\rho}_{AB}(\alpha,\eta_A,\eta_B')$ for all $\eta_B'<\eta_B$.
\end{lemma}
\begin{proof}
  This follows from Remark~\ref{rmk:convex} and the fact that
  $\hat{\rho}_{AB}(\alpha,\eta_A,\eta_B')$ is equal to $\frac{\eta_B'}{\eta_B}\hat{\rho}_{AB}(\alpha,\eta_A,\eta_B)+\frac{\eta_B-\eta_B'}{\eta_B}\hat{\rho}_A(\alpha,\eta_A,\eta_B)\ot\id/2,$
  i.e., is a convex combination of $\hat{\rho}_{AB}(\alpha,\eta_A,\eta_B)$ and
  $\hat{\rho}_A(\alpha,\eta_A,\eta_B)\ot\id/2$, both of which have LHS models.
\end{proof}
Hence, although we cannot use Theorem~\ref{steerthm2} throughout the
$\alpha$-$\eta_B$ plane, we can nevertheless establish steerability of
all states of the form $\hat{\rho}_{AB}(\alpha,1,\eta_B)$ for
$\alpha\gtrapprox 0.37$ (for example).  Furthermore, the numerics point
to the existence of a critical value around $0.92$ such that for
values of $\eta_A$ below this we can always use our criteria
(graphically, the boundary of the region in which a LHS model exists
always lies above $\eta_B=\cos (2\alpha)$ for $\eta_A\lessapprox0.92$).

\acknowledgements

We are grateful to Kim Winick for numerous helpful discussions, to
Emanuel Knill, Sania Jevtic, Stephen Jordan, and Chau Nguyen for
useful feedback on an earlier version of the manuscript, and to
Nicholas Brunner, Daniel Cavalcanti and Ivan Supic for pointers to the
literature. RC is supported by the EPSRC's Quantum Communications Hub
(grant number EP/M013472/1) and by an EPSRC First Grant (grant number
EP/P016588/1).  CAM and YS were supported in part by US NSF grants 1500095,
1526928, and 1717523. YS was also supported in part by University of Michigan.

\appendix

\section{Summary of known results for Werner states}
\label{wernersubsec}

Werner states (cf.\ \eqref{eq:werner}) are separable if and only if
$\eta\leq\frac{1}{3}$~\cite{Werner}, are steerable if
$\eta>\frac{1}{2}$~\cite{Wiseman:2007} and are non-local if
$\eta>1/K_G(3)$~\cite{AGT06}, where $K_G(3)$ is Grothendieck's
constant of order 3~\cite{Grothendieck}, which is known to satisfy
$1.426<K_G(3)<1.464$, so that
$0.683<1/K_G(3)<0.701$~\cite{Brierly,Hirsch17}.  They are local for
projective measurements if $\eta\leq 1/K_G(3)$~\cite{AGT06} and are
local for all measurements for $\eta\leq0.455$~\cite{Hirsch17} and
also have a LHS model for all measurements for $\eta\leq
5/12$~\cite{Barrett02,Quintino15}.  For $1/3<\eta\leq5/12$ the
states are non-separable and unsteerable.  For
$\frac{1}{2}<\eta\leq\frac{1}{K_G(3)}$ the states are local for
projective measurements and steerable. It is unknown whether these
states are local for all measurements anywhere in this range, which
would show steerability $\nimplies$ non-locality, however, this
non-implication is known using another family of
states~\cite{Quintino15}.

The above is summarized in
Figure~\ref{fig:wernersummary}.

\begin{figure}
\includegraphics{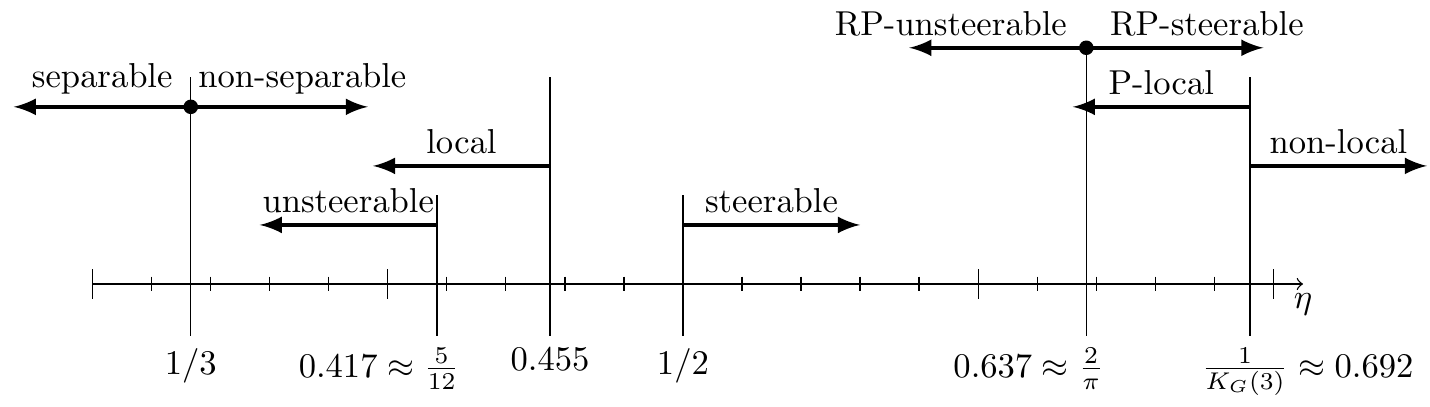}
\caption{Summary of known results for Werner states.  The
  approximation taken for $1/K_G(3)$ is the mean of the known upper
  and lower bounds.}
\label{fig:wernersummary}
\end{figure}

\section{Additional Proofs}

\subsection{Proof of Proposition~\ref{2stepprop}}\label{app:2prop}
This proof uses similar methods to those
in~\cite{Nguyen:2016}.

The proof will be divided into two cases: (1) the case where $z$ lies
on the boundary of $\B(\mu)$, and (2) the case where $z$ lies in the
interior of $\B(\mu)$.

  (1) In the case where $z$ lies on the boundary of $\B(\mu)$,
  because $\B(\mu)$ is convex, there must exist
  $H\in\mathcal{RA}(\mathbb{C}^2)$ such that the function
  $x\mapsto\left< x,H\right>$ on $\B(\mu)$ is maximized at
  $z$. We subdivide into three cases depending on $H$.

  \textit{Case 1a:} The element $z$ is on the boundary and $H>0$.

  The operator
  \begin{eqnarray}
    \rho&=&\int_{x\in\mathbb{RP}^1}x\ \mathrm{d}\mu 
  \end{eqnarray}
  is greater than or equal to $z$, so
  $\left< \rho-z,H\right>\geq 0$.  But this quantity cannot
  exceed $0$ by assumption, so $\left< \rho-z,H\right>=0$,
  which yields $\rho=z$.  Since the constant function
  $\mathbb{RP}^1\to\{1\}$ satisfies the definition of a two-step
  function, we are done.

  \textit{Case 1b:} The element $z$ is on the boundary and
  $H \ngeq 0$.

  In this case, there are unique distinct elements
  $y,w\in\mathbb{RP}^1$ such that
  $\left< y,H\right>=\left< w,H\right>=0$,
  $\left< x,H\right> >0$ for all $x\in(y,w)$, and
  $\left< x,H\right> <0$ for all $x\in(w,y)$.  Choose a
  function $f\colon\mathbb{RP}^1\to[0,1]$ such that
  \begin{eqnarray}\label{eq:z}
    z=\int_{x\in \mathbb{RP}^1}xf(x)\ \mathrm{d}\mu 
  \end{eqnarray}
  (such a function must exist because $z\in\B(\mu)$).  Let $g$ be the
  two step-function
  \begin{eqnarray}
    g(x)&=&\left\{\begin{array}{ccl} 1 & \textnormal{ if }
                            & x \in (y,w) \\
                            0 & \textnormal{ if } & x \in (w,y) \\
                            f(y) & \textnormal{ if } & x=y \\
                            f(w) & \textnormal{ if } & x=w. 
                          \end{array} \right.
  \end{eqnarray}
  and let
  \begin{eqnarray}
    r&=&\int_{x\in\mathbb{RP}^1}xg(x)\ \mathrm{d}\mu \, .
  \end{eqnarray}
  Since $r\in\B(\mu)$, $\left< r,H\right>\leq\left<z,H\right>$.
  Hence we have 
  \begin{eqnarray}
    0 \geq \left<r-z,H\right>&=&\left<\int_{x\in(y, w)}(1-f(x))x\ \mathrm{d}\mu ,H\right>-\left<\int_{x\in(w, y)}f(x)x\ \mathrm{d}\mu,H \right>\geq 0\, ,\label{zrexp}
  \end{eqnarray}
  where the final inequality follows because any operator
  $x\in(y,w)$ has positive inner product with $H$ and any operator
  $x\in(w,y)$ has negative inner product with $H$.  It follows
  that $z=r$, which completes this case.

  \textit{Case 1c:} The element $z$ is on the boundary and $H$ is
  positive semidefinite and rank-one.

  Let $y\in\mathbb{RP}^1$ be the unique element such that
  $\left<H,y\right>=0$.  Let
  \begin{eqnarray}
    g(x)&=&\left\{\begin{array}{ccl} 1 & \textnormal{ if }
                            & x\neq y \\
                            f(y) & \textnormal{ if } & x=y\, ,
                          \end{array} \right.
  \end{eqnarray}
  where $f\colon\mathbb{RP}^1\to[0,1]$ is a function such
  that~\eqref{eq:z} holds.  By similar reasoning as in Case~1b, this
  function also computes $z$.

  \textit{Case 2:} The element $z$ is in the interior of $\B(\mu)$.

  Let
  \begin{eqnarray}
    c&=&\int_{x\in\mathbb{RP}^1}(1/2)x\ \mathrm{d}\mu .
  \end{eqnarray}
  Since $z$ is interior it can be written as $z=tc+(1-t)b$,
  where $t\in[0,1]$ and $b$ is an element on the boundary of
  $\B(\mu)$.  Let $g$ be a two-step function which computes $b$,
  which must exist from the first part of the proof.  Then, the
  function $t/2+(1-t)g$ computes $z$. \qed

\subsection{Proof of Lemma~\ref{lem:leq2}}\label{app:leq2}
We will  construct an explicit curve which is the boundary
of (\ref{region1}).  Let $S=\{s_1,\ldots,s_n\}$ be the support
of $\mu$, where the points $\proj{0},s_1,\ldots,s_n$ are in
clockwise order, and define $\tilde{\rho}_m:=\sum_{i=1}^m\mu(s_i)s_i$.

For any $t\in[0,\left<H,\rho\right>]$, define a two-step function $h_t\colon\mathbb{RP}^1\to[0,1]$ as follows: if
\begin{eqnarray}
t\in\left[\left<\rho_m,H\right>,\left<\rho_{m+1},H\right>\right),
\end{eqnarray}
then
\begin{eqnarray}
h_t(x)&=&1\quad\text{for }x\in\left[\proj{0},s_{m+1}\right)\\
h_t(s_{m+1})&=&\left(\frac{t-\left<\rho_m,H\right>}{\mu(s_{m+1})\left<s_{m+1},H\right>}\right),
\end{eqnarray}
and $h_t$ is zero elsewhere. Note that, by construction,
\begin{eqnarray}
\int_{x\in\mathbb{RP}^1}h_t(x)\left<x,H\right>\ \mathrm{d}\mu&=&t.
\end{eqnarray}
Also define a zero-bias two-step function
$\overline{h}_t\colon\mathbb{RP}^1\to[0,1]$ by
\begin{eqnarray}
\overline{h}_t & = & \left\{ \begin{array}{ll} h_{\left( t + \left< \rho,H\right> /2 \right)}  - h_t   & \hskip0.3in \textnormal{ if } t < \left< \rho , H \right> /2 \\ \\
1 - h_t  + h_{\left(t - \left< \rho,H \right>/2\right)}   & \hskip0.3in \textnormal{ otherwise,}\\
\end{array} \right.
\end{eqnarray}
so that for any $t$,
\begin{eqnarray}
\int_{x \in \mathbb{RP}^1 } \overline{h}_t ( x )\left< x, H \right>\ \mathrm{d}\mu & = & \left< \rho, H \right>/2.
\end{eqnarray}
Let
\begin{eqnarray}
G ( t ) & = & \int_{x \in \mathbb{RP}^1} \overline{h}_t (x) x \ \mathrm{d} \mu.
\end{eqnarray}
The points in the image of $G ( t )$ are in the region (\ref{region1}) by construction,
and since they are obtained from zero-bias two-level functions, they lie
on the boundary of $\textnormal{Box} ( \mu )$ (see Proposition~\ref{2stepprop}).  
The image of $G$ is the boundary of (\ref{region1}).

Note that for any fixed $i$, the function $t \mapsto \overline{h}_t ( s_i )$ is
bitonic (in the sense that it only increases once and decreases once, modulo $\left< \rho , H \right>$) and thus
\begin{eqnarray}
\int_0^{\left< \rho, H \right>} \left| \deriv{t} \left(  \overline{h}_t ( s_i )  \right) \right| \mathrm{d}t \leq 2.
\end{eqnarray}
Therefore, the length of the curve $G$ satisfies
\begin{eqnarray}
\int_0^{\left< \rho, H \right>} \left\| \deriv{t} G ( t ) \right\|_1 \mathrm{d}t & = &
\int_0^{\left< \rho, H \right>} \left\| \deriv{t} \int_{x \in \mathbb{RP}^1} \overline{h}_t ( x ) x \mathrm{d} \mu  \right\|_1 \mathrm{d}t \\
& = & \int_0^{\left< \rho, H \right>} \left\| \deriv{t} \sum_{i=1}^n \overline{h}_t ( s_i ) s_i \mu ( s_i )  \right\|_1 \mathrm{d}t \\
& \leq & \int_0^{\left< \rho, H \right>} \sum_{i=1}^n \left| \deriv{t} \left( \overline{h}_t ( s_i ) \right)  \right|  \mu ( s_i )\ \mathrm{d}t \\
& \leq & \sum_{i=1}^n 2   \mu ( s_i ) = 2\, , 
\end{eqnarray}
as desired.\qed

\subsection{Tilt of the derivative of the steering ellipse}
\begin{lemma}\label{lem:perptilt}
  Suppose $\lambda,\mu\in\cR\cA(\mathbb{C}^2)$ with
  $\Tr(\lambda\mu)=0$ and $\Tilt(\mu)<1$.  Then $\Tilt(\lambda)>1$.
\end{lemma}
\begin{proof}
  Suppose $\lambda=\frac{1}{2}(n\id+r_1\sigma_1+r_3\sigma_3)$ and
  $\mu=\frac{1}{2}(m\id+s_1\sigma_1+s_3\sigma_3)$ and write
  $(r_1,r_3)=r{\bf e}_r$ and $(s_1,s_3)=s{\bf e}_s$, where ${\bf
    e}_r,{\bf e}_s$ are unit vectors and $r,s\geq 0$. 

  The condition $\Tr(\lambda\mu)=0$ can be written
  $-{\bf r}.{\bf s}=nm$.  $\Tilt(\mu)<1$ is equivalent to $s^2<m^2$.
  It follows that $({\bf r}.{\bf s})^2=n^2m^2>n^2s^2$.  This
  rearranges to $({\bf r}.{\bf e_s})^2>n^2$, from which it follows
  that $r^2>n^2$, i.e., $\Tilt(\lambda)>1$.
\end{proof}
\begin{corollary}\label{cor:app1}
  Let $\rho_{AB} \in \cR\cD ( \mathbb{C}^2 \ot \mathbb{C}^2 )$ be a
  two-qubit state whose steering ellipse $\{ \tilde{\rho}_B ( \theta )
  \}$ has tilt smaller than $1$. Then $\Tilt(\deriv{\theta}
  \tilde{\rho}_B(\theta))>1$ for all $\theta$.
\end{corollary}

\subsection{A topological lemma}

\begin{lemma}
\label{toplemma}
Let $D = \{ z \in \mathbb{C} \mid \left| z \right| \leq 1 \}$ and let $S^1 = \{ z \in \mathbb{C} \mid \left| z \right| = 1 \}$.  Let
$F \colon D \to D$ be a continuous function such that for any $z \in S^1$, $F ( z ) = -z$.  Then, $F$ is onto.
\end{lemma}

\begin{proof}
  Suppose, for the sake of contradiction, that
  $y\in D\smallsetminus F(D)$.  Let $G \colon D \to S^1$ be the
  (unique) function defined by the condition that for any $z \in D$,
  $F(z)$ lies on the line segment from $y$ to $G(z)$.  Note that the
  function $G$ also satisfies $G(z)=-z$ for $z\in S^1$.  The family of functions
  $\left\{ H_\alpha \colon S^1 \to S^1 \mid \alpha \in [ 0, 1 ]
  \right\}$
  given by $H_\alpha ( z ) = G ( \alpha z )$ is a continuous
  deformation between the negation map on $S^1$ and the constant map
  which takes $S^1$ to $G ( 0 )$.  This is impossible, since these
  maps represent different elements of the fundamental group of $S^1$.
  Thus, by contradiction, the original map $F$ must be onto.
\end{proof}

\subsection{Proof of Proposition~\ref{prop:yxcomp}}\label{app:yxcomp}

By Proposition~\ref{prop:sf}, we have
\begin{eqnarray}
\left| X \right|_Y & = & Y^{-1} \left| Y X Y \right| Y^{-1} \\
& = & \frac{ Y^{-1} ( Y X Y ) ( Y X Y ) Y^{-1} - Y^{-1} ( Y X Y ) (\hat{Y} \hat{X} \hat{Y} )
Y^{-1}}{ \left\| YXY \right\|_1 } \\
& = & \frac{ X Y^2 X - X \det ( Y ) \hat{X} \hat{Y} 
Y^{-1}}{ \left\| YXY \right\|_1 } \\
& = & \frac{ X Y^2 X - X  \hat{X} \hat{Y} 
( \det ( Y ) Y^{-1})}{ \left\| YXY \right\|_1 } \\
& = & \frac{ X Y^2 X - \det(X) \hat{Y} 
\hat{Y}}{ \left\| YXY \right\|_1 },
\end{eqnarray}
which is equal to the desired formula.

\end{document}